\documentclass[journal]{IEEEtran}

\usepackage[latin1]{inputenc}
\usepackage{times}
\usepackage{amsmath}
\usepackage{amssymb}
\usepackage{amsthm}

\usepackage{mathrsfs}
\usepackage{subcaption}
\usepackage{graphicx}
\usepackage{enumerate}
\usepackage[compress]{cite}
\usepackage{float}
\usepackage[acronym]{glossaries}
\usepackage{enumitem}

\newacronym{5G}{5G}{fifth generation}
\newacronym{6G}{6G}{sixth generation}
\newacronym{AP}{AP}{access point}
\newacronym[longplural={angles of arrival}]{AoA}{AoA}{angle of arrival}
\newacronym[longplural={angles of departure}]{AoD}{AoD}{angle of departure}
\newacronym{B5G}{B5G}{beyond 5G}
\newacronym{BCD}{BCD}{block coordinate descend}
\newacronym{BRPA}{BRPA}{balanced random pilot assignment}
\newacronym{CAP}{CAP}{compress-after-precoding}
\newacronym{CB}{CB}{conjugate beamforming}
\newacronym{CBP}{CBP}{compress-before-precoding}
\newacronym{CCDF}{CCDF}{complementary cumulative distribution function}
\newacronym{CDF}{CDF}{cumulative distribution function}
\newacronym{CoMP}{CoMP}{coordinated multipoint}
\newacronym{CPA}{CPA}{cluster-based pilot assignment}
\newacronym{CPU}{CPU}{central processing unit}
\newacronym{C-RAN}{C-RAN}{cloud radio access network}
\newacronym{CSI}{CSI}{channel state information}
\newacronym{DCPA}{DCPA}{dissimilarity cluster-based pilot assignment}
\newacronym{DL}{DL}{downlink}
\newacronym{GOPA}{GOPA}{globally optimal power allocation}
\newacronym{LS}{LS}{least-squares}
\newacronym{LOS}{LOS}{line-of-sight}
\newacronym{MIMO}{MIMO}{multiple-input multiple-output}
\newacronym{M-MIMO}{M-MIMO}{massive MIMO}
\newacronym{MMSE}{MMSE}{minimum mean square error}
\newacronym{mmWave}{mmWave}{millimeter wave}
\newacronym{MRC}{MRC}{maximal ratio combining}
\newacronym{MS}{MS}{mobile station}
\newacronym{MSE}{MSE}{mean square error}
\newacronym{MU-MIMO}{MU-MIMO}{multiuser-MIMO}
\newacronym{NCB}{NCB}{normalized conjugate beamforming}
\newacronym{NMRC}{NMRC}{normalized maximal ratio combiner}
\newacronym{NOPA}{NOPA}{normalized optimal power allocation}
\newacronym{NLOS}{NLOS}{non-line-of-sight}
\newacronym{QoS}{QoS}{quality of service}
\newacronym{rms}{rms}{root mean square}
\newacronym{RHS}{RHS}{right hand side}
\newacronym{RPA}{RPA}{random pilot assignment}
\newacronym{RRM}{RRM}{radio resource management}
\newacronym{RF}{RF}{radio frequency}
\newacronym{SINR}{SINR}{signal-to-interference-plus-noise ratio}
\newacronym{SNR}{SNR}{signal-to-noise ratio}
\newacronym{SOPA}{SOPA}{sequential optimal power allocation}
\newacronym{TDD}{TDD}{time division duplexing}
\newacronym{UDN}{UDN}{ultra dense network}
\newacronym{UL}{UL}{uplink}
\newacronym{ZF}{ZF}{zero-forcing}

\newcommand{\bs}{\boldsymbol}
\DeclareMathOperator{\SINR}{SINR}
\DeclareMathOperator{\diag}{diag}

\DeclareMathOperator{\tr}{tr}
\DeclareMathOperator{\PL}{PL}

\newtheorem{theorem}{Theorem}

\begin{document}

\title{Cell-Free Millimeter-Wave Massive MIMO Systems with Limited Fronthaul Capacity}

\author{Guillem~Femenias,~\IEEEmembership{Senior Member,~IEEE,}
        and~Felip~Riera-Palou,~\IEEEmembership{Senior Member,~IEEE}
\thanks{G Femenias and F Riera-Palou are with the Mobile Communications Group, University of the Balearic Islands, Palma 07122, Illes Balears, Spain (e-mail: \{guillem.femenias,felip.riera\}@uib.es).

This work has been funded in part by the Agencia Estatal de Investigaci\'on and Fondo Europeo de Desarrollo Regional (AEI/FEDER, UE) under project TERESA (subproject TEC2017-90093-C3-3-R), Ministerio de Econom\'ia y Competitividad (MINECO), Spain.}
%\thanks{Manuscript received September 15, 2016.}
}

%\markboth{Submitted to IEEE Access}{Femenias \textit{et al.}: Cell-Free Millimeter-Wave Massive MIMO Systems with Limited Fronthaul Capacity}

\maketitle

\begin{abstract}

Network densification, massive multiple-input multiple-output (MIMO) and millimeter-wave (mmWave) bands have recently emerged as some of the physical layer enablers for the future generations of wireless communication networks (5G and beyond). Grounded on prior work on sub-6~GHz cell-free massive MIMO architectures, a novel framework for cell-free mmWave massive MIMO systems is introduced that considers the use of low-complexity hybrid precoders/decoders while factors in the impact of using capacity-constrained fronthaul links. A suboptimal pilot allocation strategy is proposed that is grounded on the idea of clustering by dissimilarity. Furthermore, based on mathematically tractable expressions for the per-user achievable rates and the fronthaul capacity consumption, max-min power allocation and fronthaul quantization optimization algorithms are proposed that, combining the use of block coordinate descent methods with sequential linear optimization programs, ensure a uniformly good quality of service over the whole coverage area of the network. Simulation results show that the proposed pilot allocation strategy eludes the computational burden of the optimal small-scale CSI-based scheme while clearly outperforming the classical random pilot allocation approaches. Moreover, they also reveal the various existing trade-offs among the achievable max-min per-user rate, the fronthaul requirements and the optimal hardware complexity (i.e., number of antennas, number of RF chains).

\end{abstract}


\IEEEpeerreviewmaketitle

\section{Introduction}

\subsection{Motivation and previous work}

\IEEEPARstart{D}{riven} by the continuously increasing demands for high system throughput, low latency, ultra reliability, improved fairness and near-instant connectivity, \gls{5G} wireless communication networks are being standardized \cite{Shafi17} while, at the same time, insights and innovations from industry and academia are paving the road for the coming of the \gls{6G} \cite{David18}. As stated by Marzetta \emph{et al.} in \cite[Chapter 1]{Marzetta16}, there are three basic pillars at the physical layer that can be used to sustain the spectral and energy efficiencies that these networks are expected to provide: (i) employing massive \gls{MIMO}, (ii) using \gls{UDN} deployments, and (iii) exploiting new frequency bands.

Massive \gls{MIMO} systems, equipped with a large number of antenna elements, are intended to be used as \gls{MU-MIMO} arrangements in which the number of antenna elements at each \gls{AP} is much larger than the number of \glspl{MS} simultaneously served over the same time/frequency resources. The operation of massive \gls{MIMO} schemes is based on the availability of \gls{CSI} acquired through \gls{TDD} operation and the use of \gls{UL} pilot signals. Such a setting allows for very high spectral and energy efficiencies using simple linear signal processing in the form of conjugate beamforming or \gls{ZF} \cite{Marzetta10,Marzetta16}.

In \glspl{UDN}, a large number of \glspl{AP} deployed within a given coverage area cooperate to jointly transmit/receive to/from a (relatively) reduced number of \glspl{MS} thanks to the availability of high-performance low-latency fronthaul links connecting the \glspl{AP} to a central coordinating node. Coordination among \glspl{AP} can effectively control (or even eliminate) intercellular interference in an approach that was first referred to as network \gls{MIMO} \cite{Karakayali06,Gesbert10}, later led to the concept of \gls{CoMP} transmission \cite{Irmer11} and, more recently, to that of \gls{C-RAN} \cite{Checko15}. In a \gls{C-RAN}, the \glspl{AP}, which are treated as a distributed \gls{MIMO} system, are connected to a cloud-computing based \gls{CPU} in charge, among many others, of the baseband processing tasks of all \glspl{AP}. Conceptually similar to the \gls{C-RAN} architecture, but explicitly relying on assumptions specific of the massive \gls{MIMO} regime, distributed massive \gls{MIMO}-based \glspl{UDN} have been recently termed as \emph{cell-free massive \gls{MIMO}} networks \cite{Ngo15,Ngo17}. In these networks, a massive number of \glspl{AP} connected to a \gls{CPU} are distributed across the coverage area and, as in the cellular collocated massive \gls{MIMO} schemes, exploit the channel hardening and favorable propagation properties to coherently serve a large number of \glspl{MS} over the same time/frequency resources. Typically using simple linear signal processing schemes, they are claimed to provide uniformly good \gls{QoS} to the whole set of served \glspl{MS} irrespective of their particular location in the coverage area.

Since the microwave radio spectrum (from 300 MHz to 6 GHz) is highly congested, the use of massive antenna systems and network densification alone may not be sufficient to meet the \gls{QoS} demands in next generation wireless communications networks. Thus, another promising physical layer solution that is expected to play a pivotal role in \gls{5G} and beyond \gls{5G} communication systems is to increase the available spectrum by exploring new less-congested frequency bands. In particular, there has been a growing interest in exploiting the so-called \gls{mmWave} bands \cite{Rappaport13,Boccardi14,Akdeniz14,Rappaport17}. The available spectrum at these frequencies is orders of magnitude higher than that available at the microwave bands and, moreover, the very small wavelengths of \glspl{mmWave}, combined with the technological advances in low-power CMOS \gls{RF} miniaturization, allow for the integration of a large number of antenna elements into small form factors. Large antenna arrays can then be used to effectively implement \gls{mmWave} massive \gls{MIMO} schemes (see, for instance, \cite{Gao18,Busari18} and references therein) that, with appropriate beamforming, can more than compensate for the orders-of-magnitude increase in free-space path-loss produced by the use of higher frequencies.

The performance of cell-free massive \gls{MIMO} using conventional sub-6 GHz frequency bands and assuming infinite-capacity fronthaul links has been extensively studied in, for instance, \cite{Ngo17,Nayebi17,Nguyen17,Ngo18}. %The authors of these research works, using a methodology similar to that exploited in the cellular massive \gls{MIMO} literature, introduce spectral and energy efficiency analytical tools that take \gls{UL} pilot-based \gls{CSI} estimation and practical pilot allocation schemes into consideration and, most importantly, show that cell-free massive \gls{MIMO} deployments outperform the uncoordinated small cells in terms of coverage probability \cite{Ngo17}.
Cell-free massive \gls{MIMO} networks using capacity-constrained fronthaul links have also been considered in \cite{Bashar18b,Boroujerdi18} but assuming, again, the use of fully digital precoders in conventional sub-6 GHz frequency bands. %In contrast to microwave communications, signal processing in \gls{mmWave} transmissions is subject to a wide set of practical constraints that need to be accounted for when designing and analyzing \gls{mmWave}-based cell-free massive \gls{MIMO} systems. In particular,
Sub-6 GHz massive \gls{MIMO} systems are often assumed to implement a fully-digital baseband signal processing requiring a dedicated \gls{RF} chain for each antenna element. The present status of \gls{mmWave} technology, however, characterized by high-power consumption levels and high production costs, precludes the fully-digital implementation of massive \gls{MIMO} architectures, and typically forces \gls{mmWave} systems to rely on hybrid digital-analog signal processing architectures. In these hybrid transceiver architectures, a large antenna array connects to a limited number of \gls{RF} chains via high-dimensional \gls{RF} precoders, typically implemented using analog phase shifters and/or analog switches, and low-dimensional baseband digital precoders are then used at the output of the \gls{RF} chains \cite{Ayach14,Gao16,Molisch17}. The network of phase shifters connecting the array of antennas to the \gls{RF} chains determines whether the structure is fully or partially connected \cite{Park17TWC}. Thus, the assumptions, methods and analytical expressions in \cite{Ngo17,Nayebi17,Nguyen17,Ngo18,Bashar18b,Boroujerdi18} cannot by applied directly when assuming the use of \gls{mmWave} frequency bands. Despite its evident potential, as far as we know, besides \cite{Alonzo17,Alonzo18} there is no other research work on cell-free \gls{mmWave} massive \gls{MIMO} systems and, furthermore, the authors of these works did not face one of the main challenges in the implementation of cooperative \glspl{UDN}, that is, the fact that these systems require of a substantial information exchange between the \glspl{AP} and the \gls{CPU} via capacity-constrained fronthaul links. Moreover, they also considered the use of oversimplified \gls{mmWave} channel models and \gls{RF} precoding stages, without constraining the available number of \gls{RF}-chains at each \gls{AP}.

\subsection{Aim and contributions}

Motivated by the above considerations, our main aim in this paper is to address the design and performance evaluation of realistic cell-free \gls{mmWave} massive \gls{MIMO} systems using hybrid precoders and assuming the availability of capacity-constrained fronthaul links connecting the \glspl{AP} and the \gls{CPU}. The main contributions of our work can be summarized as follows:

\begin{itemize}[noitemsep,wide=0pt, leftmargin=\dimexpr\labelwidth + 2\labelsep\relax]

\item The performance of both the \gls{DL} and \gls{UL} of cell-free \gls{mmWave} massive \gls{MIMO} systems is considered with particular emphasis on the per-user rate, rather than the system sum-rate, by posing max-min fairness resource allocation problems that take into account the effects of imperfect channel estimation, power control, non-orthogonality of pilot sequences, and fronthaul capacity constraints. Instead of assuming the use of rather simple uniform quantization processes when forwarding information on the capacity-constrained fronthauls, the proposed optimization problems assume the use of large-block lattice quantization codes able to approximate a Gaussian quantization noise distribution. Optimal solutions to these problems are proposed that combine the use of block coordinate descent methods with sequential linear programs.% A low-complexity near-optimal solution is also considered for the \gls{DL} case.

\item A hybrid beamforming implementation is proposed where the \gls{RF} high-dimensionality phase shifter-based precoding/decoding stage is based on large-scale second-order statistics of the propagation channel, and hence does not need the estimation of high-dimensionality instantaneous \gls{CSI}. The low-dimensionality baseband \gls{MU-MIMO} precoding/decoding stage can then be easily implemented by standard signal processing schemes using small-scale estimated \gls{CSI}. As will be shown in the numerical results section, such a reduced complexity hybrid precoding scheme, when combined with appropriate user selection, performs very well in the fronthaul capacity-constrained \gls{UDN} \gls{mmWave}-based scenarios under consideration.

\item A suboptimal pilot allocation strategy is proposed that, based on the idea of clustering by dissimilarity, avoids the computational complexity of the optimal pilot allocation scheme. The performance of the proposed \emph{dissimilarity cluster-based pilot assignment algorithm} is compared with that of both the \emph{pure random pilot allocation approach} and the \emph{balanced random pilot strategy}.

\item For those cases in which the number of active \glspl{MS} in the network is greater than the number of available \gls{RF} chains at a particular \gls{AP}, a \gls{MS} selection algorithm is proposed that aims at maximizing the minimum average sum-energy (i.e., Frobenius norm) of the equivalent channel between the \glspl{AP} and any of the active \glspl{MS}, constrained by the fact that each \gls{AP} can only beamform to a number of \glspl{MS} less or equal than the number of available \gls{RF} chains.

\end{itemize}

\subsection{Paper organization and notational remarks}

The remainder of this paper is organized as follows. In Section \ref{sec:System_model} the proposed cell-free \gls{mmWave} massive \gls{MIMO} system is introduced. Different subsections are devoted to the description of the channel model, the large-scale and small-scale training phases, the channel estimation process, and the \gls{DL} and \gls{UL} payload transmission phases. The achievable \gls{DL} and \gls{UL} rates are presented in Section \ref{sec:Achievable_rates} and further developed in Appendices \ref{app:Appendix_1} and \ref{app:Appendix_2}. Section \ref{sec:fronthaul_capacity} is dedicated to the calculation of the capacity consumption of both the \gls{DL} and \gls{UL} fronthaul links. The pilot assignment, power allocation and quantization optimization processes are described in Sections \ref{sec:pilot_assignment} and \ref{sec:power allocation_quantization}. Numerical results and discussions are provided in Section \ref{sec:numerical_results} and, finally, concluding remarks are summarized in Section \ref{sec:Conclusion}.

\emph{Notation}: Vectors and matrices are denoted by lower-case and upper-case boldface symbols. The $q$-dimensional identity matrix is represented by $\boldsymbol{I}_q$. The operator $\det(\bs{X})$ represents the determinant of matrix $\bs{X}$, $\tr(\bs{X})$ denotes its trace, $\|\bs{X}\|_F$ is its Frobenius norm, whereas $\bs{X}^{-1}$, $\bs{X}^T$, $\bs{X}^*$ and $\bs{X}^H$ denote its inverse, transpose, conjugate and conjugate transpose (also known as Hermitian), respectively. With a slight abuse of notation, the operator $\diag(\boldsymbol{x})$ is used to denote a diagonal matrix with the entries of vector $\boldsymbol{x}$ on its main diagonal, and the operator $\diag(\boldsymbol{X})$ is used to denote a vector containing the entries in the main diagonal of matrix $\bs{X}$. The expectation operator is denoted by $\mathbb{E}\{\cdot\}$. Finally, $\mathcal{CN}(\bs{m},\bs{R})$ denotes a circularly symmetric complex Gaussian vector distributions with mean $\bs{m}$ and covariance $\bs{R}$, $\mathcal{N}(0,\sigma^2)$ denotes a real valued zero-mean Gaussian random variable with standard deviation $\sigma$, and $\mathcal{U}[a,b]$ represents a random variable uniformly distributed in the range $[a,b]$.

\section{System model}
\label{sec:System_model}

Let us consider a cell-free massive \gls{MIMO} system where a \gls{CPU} coordinates the communication between $M$ \glspl{AP} and $K$ single-antenna \glspl{MS} randomly distributed in a large area. Each of the \glspl{AP} communicates with the \gls{CPU} via error-free fronthaul links with \gls{DL} and \gls{UL} capacities ${C_F}_d$ and ${C_F}_u$, respectively. Baseband processing of the transmitted/received signals is performed at the \gls{CPU}, while the \gls{RF} operations are carried out at the \glspl{AP}. Each \gls{AP} is equipped with an array of $N > K$ antennas and $L \leq N$ \gls{RF} chains. A fully-connected architecture is considered where each \gls{RF} chain is connected to the whole set of antenna elements using $N$ analog phase shifters. Without loss of essential generality, it is assumed in this paper that the number of active \gls{RF} chains at each of the \glspl{AP} in the network is equal to $L_A=\min\{K,L\}$. That is, if $K \leq L$, all \glspl{AP} in the cell-free network provide service to the whole set of \glspl{MS} and if $K > L$, instead, each \gls{AP} can only provide service to $L$ out of the $K$ \glspl{MS} in the network and, thus, an algorithm must be devised to decide which are the \glspl{MS} to be beamformed by each of the \glspl{AP}.

\begin{figure}
  \centering
  \includegraphics[width=7.8cm]{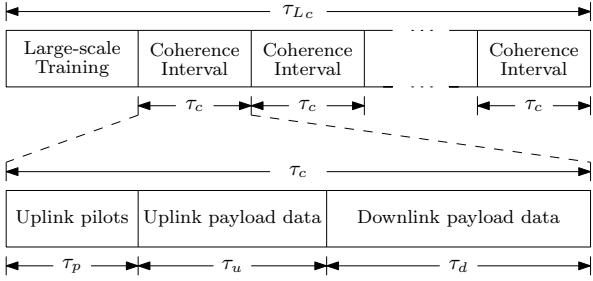}
  \caption{Allocation of the samples in large-scale and short-scale coherence intervals.}\label{fig:TDD_frame}
\end{figure}

The propagation channels linking the \glspl{AP} to the \glspl{MS} are typically characterized by small-scale parameters that are (almost) static over a coherence time-frequency interval of $\tau_c$ time-frequency samples (see \cite[Chapter 2]{Marzetta16}), and large-scale parameters (i.e., path loss propagation losses and covariance matrices) that can be safely assumed to be static over a time-frequency interval ${\tau_L}_c \gg \tau_c$. As shown in the following subsections, these channel characteristics can be leveraged to simplify both the channel estimation and the precoding/combining processes. In particular, \gls{DL} and \gls{UL} transmissions between \glspl{AP} and \glspl{MS} are organized in a half-duplex \gls{TDD} operation whereby each coherence interval is split into three phases, namely, the \gls{UL} training phase, the \gls{DL} payload data transmission phase and the \gls{UL} payload data transmission phase, and every \emph{large-scale coherence interval} ${\tau_L}_c$ the system performs an estimation of the large-scale parameters of the channel (see Fig. \ref{fig:TDD_frame}). In the \gls{UL} training phase, all \glspl{MS} transmit \gls{UL} training orthogonal pilots allowing the \glspl{AP} to estimate the propagation channels to every \gls{MS} in the network\footnote{Note that channel reciprocity can be exploited in \gls{TDD} systems and therefore only \gls{UL} pilots need to be transmitted.}. Subsequently, these channel estimates are used to detect the signals transmitted from the \glspl{MS} in the \gls{UL} payload data transmission phase and to compute the precoding filters governing the \gls{DL} payload data transmission. Not shown are guard intervals between \gls{UL} and \gls{DL} transmissions.

\subsection{Channel Model}
\label{subsec:Channel_model}

MmWave propagation is characterized by very high distance-based propagation losses that lead to sparse scattering multipath propagation. Furthermore, the use of mmWave transmitters and receivers with large tightly-packet antenna arrays results in high antenna correlation levels. These characteristics make most of the statistical channel models used in conventional sub-6 GHz \gls{MIMO} research work inaccurate when dealing with mmWave scenarios. Thus, a modified version of the discrete-time narrowband clustered channel model proposed by Akdeniz \emph{et al.} in \cite{Akdeniz14} and further extended by Samimi and Rappaport in \cite{Samimi14} will be used in this paper to capture the peculiarities of mmWave channels.

The link between the $m$th \gls{AP} and the $k$th \gls{MS} will be considered to be in one out of three possible conditions: outage, \gls{LOS} or \gls{NLOS} with probabilities:
\begin{subequations}
\begin{equation}
   p_{out}(d_{mk})=\max\left(0,1-e^{-a_{out}d_{mk}+b_{out}}\right),
\end{equation}
\begin{equation}
   p_{\gls{LOS}}(d_{mk})=\left(1-p_{out}(d_{mk})\right) e^{-a_{\gls{LOS}}d_{mk}},
\end{equation}
\begin{equation}
   p_{\gls{NLOS}}(d_{mk})=1-p_{out}(d_{mk})-p_{\gls{LOS}}(d_{mk}),
\end{equation}
\end{subequations}
respectively, where $d_{mk}$ is the distance (in meters) between the \gls{AP} and the \gls{MS}, and, according to \cite[Table I]{Akdeniz14}, $1/a_{out}=30$~m, $b_{out}=5.2$, and $1/a_{\gls{LOS}}=67.1$~m. Those links that are in outage will be characterized with infinite propagation losses, while for the links that are not in outage, the propagation losses will be characterized using a standard linear model with shadowing as
\begin{equation}
   \PL(d_{mk})[dB]=\alpha+10 \beta \log_{10}(d_{mk})+\chi_{mk},
\end{equation}
where $\alpha$ and $\beta$ are the least square fits of floating intercept and slope and depend on the carrier frequency and on whether the link is in \gls{LOS} or \gls{NLOS} (see \cite[Table I]{Akdeniz14}), and $\chi_{mk}$ denotes the large-scale shadow fading component, which is modelled as a zero mean spatially correlated normal random variable with standard deviation $\sigma_\chi$ (again, see \cite[Table I]{Akdeniz14} to obtain the typical values of $\sigma_\chi$ for \gls{LOS} and \gls{NLOS} links) whose spatial correlation model is described in \cite[(54)-(55)]{Ngo17}.

The \gls{UL} channel vector $\bs{h}_{mk} \in \mathbb{C}^{N \times 1}$ between \gls{MS} $k$ and \gls{AP} $m$ will be modelled as the sum of the contributions of $C_{mk}$ scattering clusters, each contributing $P_{mk}$ propagation paths as
\begin{equation}
   \bs{h}_{mk}=\sum_{c=1}^{C_{mk}}\sum_{p=1}^{P_{mk}} \alpha_{mk,cp} \bs{a}\left(\theta_{mk,cp},\phi_{mk,cp}\right),
\end{equation}
where $\alpha_{mk,cp}$ is the complex small-scale fading gain on the $p$th path of cluster $c$, and $\bs{a}\left(\theta_{mk,cp},\phi_{mk,cp}\right)$ represents the \gls{AP} normalized array response vector at the azimuth and elevation angles $\theta_{mk,cp}$ and $\phi_{mk,cp}$, respectively. These angles, as stated by Akdeniz \emph{et al.} in \cite[Section III.E]{Akdeniz14} can be generated as wrapped Gaussians around the cluster central angles with standard deviation given by the \gls{rms} angular spreads for the cluster. The azimuth cluster central angles are uniformly distributed in the range $[-\pi,\pi]$ and the elevation cluster central angles are set to the \gls{LOS} elevation angle. Moreover, the cluster \gls{rms} angular spreads are exponentially distributed with a mean equal to $1/\lambda_{\gls{rms}}$ that depends on the carrier frequency and on whether we are considering the azimuth or elevation directions (see \cite[Table I]{Akdeniz14}). The number of clusters is distributed as a random variable of the form
\begin{equation}
   C_{mk} \sim \max\left\{\text{Poisson}(\sigma_C),1\right\},
\end{equation}
where $\sigma_C$ is set to the empirical mean of $C_{mk}$. The small-scale fading gains are distributed as
\begin{equation}
   \alpha_{mk,cp} \sim \mathcal{CN}\left(0,\gamma_{mk,c}10^{-\PL(d_{mk})/10}\right),
\end{equation}
where the cluster $c$ is assumed to contribute with a fraction of power given by
\begin{equation}
   \gamma_{mk,c}=\frac{N \gamma'_{mk,c}}{P_{mk}\sum_{j=1}^{C_{mk}} \gamma'_{mk,j}},
\end{equation}
with
\begin{equation}
   \gamma'_{mk,j}=U_{mk,j}^{r_\tau-1} 10^{Z_{mk,j}/10},
\end{equation}
$U_{mk,j} \sim \mathcal{U}[0,1]$, $Z_{mk,j} \sim \mathcal{N}(0,\zeta^2)$, and the constants $r_\tau$ and $\zeta^2$ being treated as model parameters (see \cite[Table I]{Akdeniz14}).

Although the small-scale fading gains $\alpha_{mk,cp}$ are assumed to be static throughout the coherence interval and then change independently (i.e., block fading), the spatial covariance matrices
\begin{equation}
\begin{split}
   \bs{R}_{mk}=&\mathbb{E}\left\{\bs{h}_{mk} \bs{h}_{mk}^H \right\} \\
              =&10^{-\PL(d_{mk})/10}\sum_{c=1}^{C_{mk}}\gamma_{mk,c} \\
               &\times \sum_{p=1}^{P_{mk}} \bs{a}\left(\theta_{mk,cp},\phi_{mk,cp}\right)\bs{a}^H\left(\theta_{mk,cp},\phi_{mk,cp}\right),
\end{split}
\end{equation}
are assumed to vary at a much smaller pace (i.e., ${\tau_L}_c \gg \tau_c$).

\subsection{Large-scale training phase}
\label{subsec:large-scale-training}

\subsubsection{\gls{RF} precoder/combiner design}

In order to exploit the \gls{UL}/\gls{DL} channel reciprocity using the \gls{TDD} frame structure shown in Fig. \ref{fig:TDD_frame}, it is assumed in this paper that the $N \times L_A$ \gls{RF} matrix $\bs{W}_m^{RF}$, describing the effects of the active analog phase shifters at the $m$th \gls{AP}, is common to the \gls{DL} (\gls{RF} precoding phase) and \gls{UL} (\gls{RF} combining phase). Furthermore, denoting by $\mathcal{K}_m=\left\{\kappa_{m 1}, \ldots, \kappa_{m L_A}\right\}$ the set of $L_A$ \glspl{MS} beamformed by the $m$th \gls{AP}, it is assumed that $\bs{W}_m^{RF}$ is a function of only the spatial channel covariance matrices $\left\{\bs{R}_{mk}\right\}_{k\in\mathcal{K}_m}$, known at the $m$th \gls{AP} through spatial channel covariance estimation for hybrid analog-digital \gls{MIMO} precoding architectures (see e.g. \cite{Adhikary14,Mendez15,Park16,Park17}).

Using eigen-decomposition, the covariance matrix of the propagation channel linking \gls{MS} k and \gls{AP} $m$ can be expressed as $\bs{R}_{mk}=\bs{U}_{mk}\bs{\Lambda}_{mk}\bs{U}_{mk}^H$, where $\bs{\Lambda}_{mk}=\diag\left(\left[\lambda_{mk,1}\,\ldots\,\lambda_{mk,r_{mk}}\right]\right)$ contains the $r_{mk}$ non-null eigenvalues of $\bs{R}_{mk}$, and $\bs{U}_{mk}$ is the $N \times r_{mk}$ matrix of the corresponding eigenvectors. Hence, assuming the use of (constrained) statistical eigen beamforming \cite{Park17TSC,Mai18}, the analog \gls{RF} precoder/combiner can be designed as
\begin{equation}
\begin{split}
   \bs{W}_m^{RF}=&\begin{bmatrix}
                    \bs{w}_{m \kappa_{m 1}}^{RF} & \ldots & \bs{w}_{m \kappa_{m L_A}}^{RF}
                 \end{bmatrix}\\
                =&\begin{bmatrix}
                    e^{-j\angle\bs{u}_{m \kappa_{m 1},\max}} & \ldots & e^{-j\angle\bs{u}_{m \kappa_{m L_A},\max}}
                 \end{bmatrix},
\end{split}
\end{equation}
where $\bs{u}_{mk,\max}$ is the dominant eigenvector of $\bs{R}_{mk}$ associated to the maximum eigenvalue $\lambda_{mk,\max}$, and the function $\angle\bs{x}$ returns the phase angles, in radians, for each element of the complex vector $\bs{x}$. Note that using the \gls{RF} precoding/combining matrix, the equivalent channel vector between \gls{MS} k and \gls{AP} $m$, including the \gls{RF} precoding/decoding matrix, is defined as
\begin{equation}
   \bs{g}_{mk} = {\bs{W}_m^{RF}}^T \bs{h}_{mk}\ \in\ \mathbb{C}^{L_A \times 1},
\end{equation}
whose dimension is much less than the number of antennas of the massive \gls{MIMO} array used at the $m$th \gls{AP}, thus largely simplifying the small-scale training phase.

\subsubsection{Selection of \glspl{MS} to beamform from each \gls{AP}}

As previously stated, in those highly probable cases in which the number of active \glspl{MS} in the network is greater than the number of available \gls{RF} chains at each \gls{AP} (i.e., $K > L$), the $m$th \gls{AP}, with $m\in\{1,\ldots,M\}$, can only beamform to a group of $L$ out of the $K$ \glspl{MS} in the network, which are indexed by the set $\mathcal{K}_m=\left\{\kappa_{m1},\ldots,\kappa_{mL}\right\}$. As the \gls{RF} beamforming matrices at the \glspl{AP} are a function of only the large-scale spatial channel covariance matrices and are common to both the \gls{UL} and the \gls{DL}, the selection of the sets of \glspl{MS} to beamform from each \gls{AP} must also be based only on the available large-scale \gls{CSI}. Inspired by the Frobenius norm-based suboptimal user selection algorithm proposed by Shen \emph{et al.} in \cite{Shen06}, a selection algorithm is proposed that aims at maximizing the sum of the average energy (i.e., average Frobenius norm) of the equivalent channels (including the corresponding beamformer) between the $M$ \glspl{AP} and the $K$ \glspl{MS} with the constraints that, first, the minimum average energy of the equivalent channel between the $M$ \glspl{AP} and any of the active \glspl{MS} must be maximized and, second, that each \gls{AP} can only beamform to $L$ \glspl{MS}. Note that this optimization problem, which tends to provide some degree of (average) max-min fairness among \glspl{MS}, can be efficiently solved by using an iterative reverse-delete algorithm (similar to that used in graph theory to obtain a minimum spanning tree from a given connected, edge-weighted graph). In particular, at the beginning of the $i$th iteration of the algorithm the cell-free network is represented by a very simple edge-weighted directed graph with $M$ source nodes and $K$ sink nodes, where the $m$th source node, representing the $m$th \gls{AP}, is connected to a group $\mathcal{K}_m^{(i)}$ of sink nodes, representing the active \glspl{MS} beamformed by the $m$th \gls{AP}. The connection (edge) between the $m$th source node and the $l$th sink node in $\mathcal{K}_m^{(i)}$ is weighted by the average Frobenius norm of the equivalent channel linking the $m$th \gls{AP} and \gls{MS} $l\in \mathcal{K}_m^{(i)}$, that can be obtained as
\begin{equation}
   \xi_{ml}=\mathbb{E}\left\{\left\|{\bs{w}_{ml}^{RF}}^T\bs{h}_{ml}\right\|_F^2\right\}={\bs{w}_{ml}^{RF}}^T \bs{R}_{ml} \bs{w}_{ml}^{RF}.
\end{equation}
The average sum energy of the equivalent channels between the $M$ \glspl{AP} and \gls{MS} $k$ at the beginning of the $i$th iteration is
\begin{equation}
   \mathcal{E}_k^{(i)} = \sum_{m\in\mathcal{M}_k^{(i)}} \xi_{mk},
\end{equation}
where $\mathcal{M}_k^{(i)}$ is the set of \glspl{AP} beamforming to \gls{MS} $k$ at the beginning of the $i$th iteration. During this iteration, the reverse-delete algorithm removes the edge (i.e., the \gls{RF} chain and associated beamformer) that, first, goes out of one of those \glspl{AP} still beamforming to more than $L$ \glspl{MS} and, second, has the minimum weight maximizing the minimum average sum energy after removal. The algorithm begins with a fully connected graph and stops when all \glspl{AP} beamform to exactly $L$ \glspl{MS}. Hence, note that $M(K-L)$ iterations are needed to select the sets $\mathcal{K}_m$ for $m\in\{1,\ldots,M\}$.

\subsection{Small-scale training phase}
\label{subsec:small-scale-training}
Communication in any coherence interval of a \gls{TDD}-based massive \gls{MIMO} system invariably starts with the \glspl{MS} sending the pilot sequences to allow the channel to be estimated at the \glspl{AP}. Let $\tau_p$ denote the \gls{UL} training phase duration (measured in samples on a time-frequency grid) per coherence interval. During the \gls{UL} training phase, all $K$ \glspl{MS} simultaneously transmit pilot sequences of $\tau_p$ samples to the \glspl{AP} and thus, the $L_A \times \tau_p$ received \gls{UL} signal matrix at the $m$th \gls{AP} is given by
\begin{equation}
   {\bs{Y}_p}_m=\sqrt{\tau_p P_p}\sum_{k'=1}^K \bs{g}_{mk'} \bs{\varphi}_{k'}^T+{\bs{N}_p}_m,
\end{equation}
where $P_p$ is the transmit power of each pilot symbol, $\bs{\varphi}_k$ denotes the $\tau_p\times 1$ training sequence assigned to \gls{MS} $k$, with $\|\bs{\varphi}_k\|_F^2=1$, and ${\bs{N}_p}_m$ is an $L_A\times\tau_p$ matrix of i.i.d. additive noise samples with each entry distributed as\footnote{Note that in the \gls{UL} of a fully-connected hybrid beamforming architecture each reception chain is composed of $N$ antenna elements, each connected to a low-noise amplifier (LNA) characterized by a power gain $G_{\textrm{LNA}}$ and a noise temperature $T_{\text{LNA}}$. Each of the $N$ LNAs feeds an analog passive phase shifter characterized by an insertion loss $L_{\text{PS}}$. The outputs of the $N$ phase shifters are introduced to a power combiner whose insertion losses are typically proportional to the number of inputs, that is, $L_{\text{PC}}=N L_{\text{PC}_{in}}$. Finally, the output of the power combiner is introduced to an $\gls{RF}$ chain characterized by a power gain $G_{\text{RF}}$ and a noise temperature $T_{\text{RF}}$. Thus, the equivalent noise temperature of each receive chain can be obtained as $T_u=N \left(T_0+T_{\text{LNA}}+\frac{T_0(L_{\text{PS}}L_{\text{PC}_{in}}-1)}{G_{\text{LNA}}}+\frac{T_{\text{RF}}L_{\text{PS}}L_{\text{PC}_{in}}}{G_{\text{LNA}}}\right)$.}  $\mathcal{CN}(0,\sigma_u^2(N))$. Ideally, training sequences should be chosen to be mutually orthogonal, however, since in most practical scenarios it holds that $K>\tau_p$, a given training sequence is assigned to more than one \gls{MS}, thus resulting in the so-called pilot contamination, a widely studied phenomenon in the context of collocated massive \gls{MIMO} systems \cite{Elijah16}.

\subsection{Channel estimation}
\label{channel_estimation}

Channel estimation is known to play a central role in the performance of massive \gls{MIMO} schemes \cite{Lu14} and also in the specific context of cell-free architectures \cite{Ngo17}. %Conventional \gls{LS} estimation relies on correlating the received pilot signals with the known pilot sequences. Hence, an LS estimator of the channel vector between the $k$th \gls{MS} and the $m$th \gls{AP} can be obtained as
%\begin{equation}
%\begin{split}
%   &{\bs{g}_{LS}}_{mk} = {\bs{Y}_p}_m \bs{\varphi}_k^*  = \sqrt{\tau_p P_p} \bs{g}_{mk} \\
%   &\quad + \sqrt{\tau_p P_p}\sum_{k'\neq k} \bs{g}_{mk'} \bs{\varphi}_{k'}^T\bs{\varphi}_k^* + {\bs{N}_p}_m \bs{\varphi}_k^*,
%\end{split}
%\label{eq:hathLSmk}
%\end{equation}
%where it can be clearly observed that the estimated channel will be affected by pilot contamination whenever any of the \glspl{MS} $k' \neq k$ is allocated the same pilot sequence as \gls{MS} $k$. With the aim of reducing the pilot contamination effect, side information lying in the covariance matrices (second order statistics) of the channel vectors can be exploited by a \gls{MMSE} estimator to improve the quality of the channel estimation.
The \gls{MMSE} estimation filter for the channel between the $k$th active \gls{MS} and the $m$th \gls{AP} can be calculated as
\begin{equation}
\begin{split}
   \bs{D}_{mk}&=\arg\min_{\bs{D}}\mathbb{E}\left\{\left\|\bs{g}_{mk}-\bs{D} {\bs{Y}_p}_m \bs{\varphi}_k^*\right\|^2\right\} \\
                          &= \sqrt{\tau_p P_p} \bs{R}_{mk}^{RF} \bs{Q}_{mk}^{-1},
\end{split}
\end{equation}
where
\begin{equation}
   \bs{R}_{mk}^{RF} = \mathbb{E}\left\{\bs{g}_{mk}  \bs{g}_{mk}^H\right\}={\bs{W}_m^{RF}}^T  \bs{R}_{mk} {\bs{W}_m^{RF}}^*,
\end{equation}
and
\begin{equation}
   \bs{Q}_{mk}=\tau_p P_p\sum_{k'=1}^K \bs{R}_{mk'}^{RF} \left|\bs{\varphi}_{k'}^T \bs{\varphi}_k^*\right|^2 + \sigma_u^2(N) \bs{I}_{L_A}.
\end{equation}
Hence, the corresponding estimated channel vector can be expressed as
\begin{equation}
   \hat{\bs{g}}_{mk}=\bs{D}_{mk} {\bs{Y}_p}_m \bs{\varphi}_k^* = \sqrt{\tau_p P_p} \bs{R}_{mk}^{RF} \bs{Q}_{mk}^{-1} {\bs{Y}_p}_m \bs{\varphi}_k^*.
\end{equation}
%As it can be observed, the MMSE channel estimations corresponding to \glspl{MS} using the same pilot sequence in different cells only differ in a multiplicative constant. Consequently, for later convenience, the estimated channel vectors $\hat{\bs{g}}_{l,l',k}$ can be expressed in terms of $\hat{\bs{g}}_{l,l,k}$ as
%\begin{equation}
%\begin{split}
%   \hat{\bs{g}}_{l,l',k}&=\frac{\beta_{l,l',k}}{\beta_{l,l,k}}\hat{\bs{g}}_{l,l,k}.
%\end{split}
%\end{equation}
The \gls{MMSE} channel vector estimates can be shown to be distributed as ${\hat{\bs{g}}}_{mk} \sim \mathcal{CN}\left(\bs{0},{\hat{\bs{R}}}_{mk}^{RF}\right)$, where
\begin{equation}
   \hat{\bs{R}}_{mk}^{RF}\triangleq \tau_p P_p \bs{R}_{mk}^{RF}\bs{Q}_{mk}^{-1}{\bs{R}_{mk}^{RF}}^H.
\end{equation}
Furthermore, the channel vector $\bs{g}_{mk}$ can be decomposed as $\bs{g}_{mk}=\hat{\bs{g}}_{mk}+\tilde{\bs{g}}_{mk}$, where $\tilde{\bs{g}}_{mk}$ is the \gls{MMSE} channel estimation error, which is statistically independent of both $\bs{g}_{mk}$ and $\hat{\bs{g}}_{mk}$.

\subsection{Downlink payload data transmission}

Let us define $\bs{s}_d=\left[{s_d}_1 \ldots {s_d}_K\right]^T$ as the $K \times 1$ vector of symbols jointly (cooperatively) transmitted from the \glspl{AP} to the \glspl{MS}, such that $E\left\{\bs{s}_d\bs{s}_d^H\right\}=\bs{I}_K$. Let us also define $\bs{x}_m=\mathcal{P}_m\left(\bs{s}_d\right)$ as the $N \times 1$ vector of signals transmitted from the $m$th \gls{AP}, where $\mathcal{P}_m\left(\bs{s}_d\right)$ is used to denote the mathematical operations (linear and/or non-linear) used to obtain $\bs{x}_m$ from $\bs{s}_d$. Note that this vector must comply with a power constraint $\mathbb{E}\left\{\left\|\bs{x}_m\right\|_F^2\right\}\leq \overline{P}_m$, where $\overline{P}_m$ is the maximum average transmit power available at \gls{AP} $m$. Using this notation, the signal received by \gls{MS} $k$ can be expressed as
\begin{equation}
   {y_d}_k=\sum_{m=1}^M \bs{h}_{mk}^T \bs{x}_m + {n_d}_k,
\label{eq:yk}
\end{equation}
where ${n_d}_k \sim \mathcal{CN}(0,\sigma_d^2)$ is the Gaussian noise sample at \gls{MS} $k$. The vector $\bs{y}_d = \left[{y_d}_1\, \ldots\, {y_d}_K\right]^T$ containing the signals received by the $K$ scheduled \glspl{MS} in the network can then be expressed as
\begin{equation}
   \bs{y}_d=\sum_{m=1}^M \bs{H}_m^T \bs{x}_m + \bs{n}_d,
\end{equation}
where $\bs{H}_m=\left[\bs{h}_{m1}\, \ldots\, \bs{h}_{mK}\right]$ and $\bs{n}_d=\left[{n_d}_1\, \ldots\, {n_d}_K\right]^T$.

The mathematical operations that symbol vector $\bs{s}_d$ undergoes before being transmitted, generically represented as $\bs{x}_m = \mathcal{P}_m(\bs{s}_d)$, for all $m \in \{1,\ldots,M\}$, include, first, a baseband precoding task at the \gls{CPU}, second, a compressing process of all or part of the data that must be sent from the \gls{CPU} to the \glspl{AP} through the fronthaul links and, third, an \gls{RF} precoding task at each of the \glspl{AP}. Let us denote by ${\mathcal{Q}_d}_m(\bs{x})$ and ${\mathcal{Q}_d}_m^{-1}(\bs{x})$ the quantization and unquantization mathematical operations performed by the \gls{CAP}-based \gls{CPU}-\gls{AP} functional split on a vector of signal samples $\bs{x}$ to be transmitted by the $m$th \gls{AP}. Due to the distortion introduced by the quantization/unquantization processes, we have that \cite{Zamir96,Femenias18}
\begin{equation}
   \hat{\mathcal{Q}}_{dm}(\bs{x})\triangleq {\mathcal{Q}_d}_m^{-1}({\mathcal{Q}_d}_m(\bs{x}))=\bs{x}+{\bs{q}_d}_m,
\end{equation}
where ${\bs{q}_d}_m$ is the quantization noise vector, which is assumed to be statistically distributed as ${\bs{q}_d}_m \sim \mathcal{CN}\left(\bs{0},\sigma_{q_{dm}}^2 \bs{I}\right)$. As shown by Zamir \emph{et al.} in \cite{Zamir96}, this assumption is supported by the fact that large-block lattice quantization codes are able to approximate a Gaussian quantization noise distribution. Thus, the mathematical operations describing the \gls{CPU}-\gls{AP} functional split considered in this paper can be summarized as
\begin{equation}
\begin{split}
   \bs{x}_m=\mathcal{P}_m(\bs{s}_d)&=\bs{W}_m^{RF} \hat{\mathcal{Q}}_{dm}\left(\bs{W}_{d\,m}^{BB} \bs{\Upsilon}^{1/2} \bs{s}_d\right) \\
                                   &=\bs{W}_m^{RF} \left(\bs{W}_{d\,m}^{BB} \bs{\Upsilon}^{1/2} \bs{s}_d + {\bs{q}_d}_m\right),
\end{split}
\label{eq:bsxm}
\end{equation}
where $\bs{W}_d^{BB} = \left[{\bs{W}_{d\,1}^{BB}}^T\ \ldots\ {\bs{W}_{d\,M}^{BB}}^T\right]^T\ \in\ \mathbb{C}^{M L_A \times K}$, with $\bs{W}_{d\,m}^{BB}=\left[\bs{w}_{dm1}^{BB}\ \ldots\ \bs{w}_{dmK}^{BB}\right]\ \in\ \mathbb{C}^{L_A \times K}$ denoting the baseband precoding matrix affecting the signal transmitted by the $m$th \gls{AP}, and $\bs{\Upsilon}=\diag\left([\upsilon_1\,\ldots\,\upsilon_K]\right)$ is a $K \times K$ diagonal matrix containing the power control coefficients in its main diagonal, which are chosen to satisfy the following necessary power constraint at the $m$th \gls{AP}
\begin{equation}
\begin{split}
   \mathbb{E}\left\{\left\|\bs{x}_m\right\|_F^2\right\} &= \sum_{k=1}^K \upsilon_k \theta_{mk}^{BB/RF}+{\sigma_q^2}_{dm} \left\|\bs{W}_m^{RF}\right\|_F^2 \\
                                                        &= \sum_{k=1}^K \upsilon_k \theta_{mk}^{BB/RF} + {\sigma_q^2}_{dm} L_A N \leq \overline{P}_m,
\end{split}
\label{eq:Power_constraint}
\end{equation}
where we have used the definition
\begin{equation}
   \theta_{mk}^{BB/RF}=\mathbb{E}\left\{\left\|\bs{W}_m^{RF} \bs{w}_{dmk}^{BB}\right\|_F^2\right\}.
\end{equation}

Using the proposed hybrid \gls{CAP} approach, the signal received by the $K$ \glspl{MS} can be rewritten as
\begin{equation}
\begin{split}
   \bs{y}_d=&\sum_{m=1}^M \bs{H}_m^T \bs{W}_m^{RF} \bs{W}_{d\,m}^{BB} \bs{\Upsilon}^{1/2} \bs{s}_d \\
            &+ \sum_{m=1}^M \bs{H}_m^T \bs{W}_m^{RF} {\bs{q}_d}_m + \bs{n}_d \\
           =&\ \bs{G}^T \bs{W}_d^{BB}\bs{\Upsilon}^{1/2} \bs{s}_d + \bs{\eta}_d,
\end{split}
\label{eq:yd}
\end{equation}
where $\bs{G}=[\bs{G}_1^T\,\ldots\,\bs{G}_M^T]^T$, with $\bs{G}_m={\bs{W}_m^{RF}}^T \bs{H}_m$, representing the equivalent \gls{MIMO} channel matrix between
the $K$ \glspl{MS} and the $M$ \glspl{AP}, including the RF precoding/decoding matrices, and
\begin{equation}
   \bs{\eta}_d  =\bs{G}^T \bs{q}_d + \bs{n}_d,
\end{equation}
with $\bs{q}_d=[{\bs{q}_d}_1^T \ldots {\bs{q}_d}_M^T]^T$, includes the thermal noise as well as the quantization noise samples received from all the \glspl{AP} in the network. Now, using the classical \gls{ZF} \gls{MU-MIMO} baseband precoder to harness the spatial multiplexing, we have that
\begin{equation}
   \bs{W}_d^{BB}=\hat{\bs{G}}^*\left(\hat{\bs{G}}^T \hat{\bs{G}}^*\right)^{-1}
\end{equation}
or, equivalently,
\begin{equation}
   \bs{W}_{d\,m}^{BB} = \hat{\bs{G}}_m^*\left(\hat{\bs{G}}^T \hat{\bs{G}}^*\right)^{-1}\ \forall m,
\end{equation}
where we have assumed that $\bs{G}=\hat{\bs{G}}+\tilde{\bs{G}}$ and $\bs{G}_m=\hat{\bs{G}}_m+\tilde{\bs{G}}_m$. Consequently, the signal received by the $k$th \gls{MS} can be expressed as
\begin{equation}
\begin{split}
   {y_d}_k = &\bs{g}_k^T \hat{\bs{G}}^*\left(\hat{\bs{G}}^T \hat{\bs{G}}^*\right)^{-1} \bs{\Upsilon}^{1/2}\bs{s}_d + {\eta_d}_k \\
           = &\left(\hat{\bs{g}}_k^T + \tilde{\bs{g}}_k^T\right) \hat{\bs{G}}^*\left(\hat{\bs{G}}^T \hat{\bs{G}}^*\right)^{-1} \bs{\Upsilon}^{1/2}\bs{s}_d + {\eta_d}_k \\
           = &\sqrt{\upsilon_k} {s_d}_k + \tilde{\bs{g}}_k^T \hat{\bs{G}}^*\left(\hat{\bs{G}}^T \hat{\bs{G}}^*\right)^{-1} \bs{\Upsilon}^{1/2}\bs{s}_d + {\eta_d}_k
\end{split}
\label{eq:ydk}
\end{equation}
where ${\eta_d}_k = \bs{g}_k^T \bs{q}_d + {n_d}_k$. The first term denotes the useful received signal, the second term contains the interference terms due to the use of imperfect \gls{CSI} (pilot contamination), and the third term encompass both the quantification and thermal noise samples.

\subsection{Uplink payload data transmission}

In the \gls{UL}, the vector of received signals at the output of the $L_A$ \gls{RF} chains (including the \gls{RF} phase shifters) of the $m$th \gls{AP} is given by
\begin{equation}
\begin{split}
   {\bs{r}_u}_m=&\sqrt{P_u}\sum_{k'=1}^K \bs{g}_{mk'} \sqrt{\omega_{k'}} {s_u}_{k'} + {\bs{n}_u}_m \\
               =&\sqrt{P_u}\bs{G}_m \bs{\Omega}^{1/2} \bs{s}_u + {\bs{n}_u}_m,
\end{split}
\end{equation}
where $P_u$ is the maximum average \gls{UL} transmit power available at any of the active \glspl{MS}, $\bs{s}_u=[{s_u}_1\,\ldots\,{s_u}_K]^T$ denotes the vector of symbols transmitted by the $K$ active \gls{MS}, $\bs{\Omega}=\diag([\omega_1\,\ldots\,\omega_K])$, with $0 \leq \omega_k \leq 1$, is a matrix containing the power control coefficients used at the \glspl{MS}, and ${\bs{n}_u}_m \sim \mathcal{CN}(\bs{0},\sigma_u^2(N) \bs{I}_{L_A})$ is the vector of additive thermal noise samples at the output of the $L_A$ \gls{RF} chains of the $m$th \gls{AP}. The received vector of signals at each of the \glspl{AP} in the network is quantized and forwarded to the \gls{CPU} via the \gls{UL} fronthaul links, where they are unquantized and jointly processed using a set of baseband combining vectors. Using a similar approach to that employed to model the \gls{DL} transmission, the received vector of (unquantized) samples from the $m$th \gls{AP} can be expressed as
\begin{equation}
   {\bs{z}_u}_m=\hat{\mathcal{Q}}_{um}\left({\bs{r}_u}_m\right)={\bs{r}_u}_m + {\bs{q}_u}_m,
\end{equation}
where ${\bs{q}_u}_m$ is the quantization noise vector, which is assumed to be statistically distributed as ${\bs{q}_u}_m \sim \mathcal{CN}\left(\bs{0},\sigma_{q_{u m}}^2 \bs{I}_{L_A}\right)$. Now, assuming the use of \gls{ZF} \gls{MIMO} detection, the \gls{CPU} uses the detection matrix
\begin{equation}
   \bs{W}_u^{BB}=\left(\hat{\bs{G}}^H \hat{\bs{G}}\right)^{-1}\hat{\bs{G}}^H = {\bs{W}_d^{BB}}^T
\end{equation}
or, equivalently
\begin{equation}
   \bs{W}_{u m}^{BB}=\left(\hat{\bs{G}}^H \hat{\bs{G}}\right)^{-1}\hat{\bs{G}}_m^H = {\bs{W}_{d m}^{BB}}^T,\ \forall m,
\end{equation}
to jointly process the vector $\bs{z}_u=\left[{\bs{z}_u}_1^T\,\ldots\,{\bs{z}_u}_M^T\right]^T$ and obtain the vector of detected samples
\begin{equation}
\begin{split}
   \bs{y}_u=&\bs{W}_u^{BB}\bs{z}_u=\sqrt{P_u}\bs{W}_u^{BB} \bs{G} \bs{\Omega}^{1/2} \bs{s}_u + \bs{\eta}_u \\
           =&\sqrt{P_u} \bs{\Omega}^{1/2} \bs{s}_u + \sqrt{P_u}\bs{W}_u^{BB} \tilde{\bs{G}} \bs{\Omega}^{1/2} \bs{s}_u + \bs{\eta}_u,
\end{split}
\end{equation}
where $\bs{\eta}_u=\bs{W}_u^{BB}\left(\bs{q}_u + \bs{n}_u\right)$. Again, the first term denotes the useful received signal, the second term contains the interference terms due to the use of imperfect \gls{CSI}, and the third term includes both the quantification and thermal noise samples. The detected sample corresponding to the symbol transmitted by the $k$th \gls{MS} can then be obtained as
\begin{equation}
   {y_u}_k =\sqrt{P_u} \omega_k^{1/2} {s_u}_k + \sqrt{P_u}\left[\bs{W}_u^{BB} \tilde{\bs{G}} \bs{\Omega}^{1/2} \bs{s}_u\right]_k + {\eta_u}_k,
\label{eq:yuk}
\end{equation}
where $[\bs{x}]_k$ denotes the $k$th entry of vector $\bs{x}$.

\section{Achievable rates}
\label{sec:Achievable_rates}

Analysis techniques similar to those applied, for instance, in \cite{Hassibi03,Yang13,Interdonato16,Marzetta16,Ngo17,Nayebi17}, are used in this section to derive \gls{DL} and \gls{UL} achievable rates. In particular, the sum of the second and third terms on the \gls{RHS} of \eqref{eq:ydk}, for the \gls{DL} case, and \eqref{eq:yuk}, for the \gls{UL} case, are treated as \emph{effective noise}. The additive terms constituting the \emph{effective noise} are, in both \gls{DL} and \gls{UL} cases, mutually uncorrelated, and uncorrelated with ${s_d}_k$ and ${s_u}_k$, respectively. Therefore, both the desired signal and the so-called \emph{effective noise} are uncorrelated. Now, recalling the fact that uncorrelated Gaussian noise represents the worst case, from a capacity point of view, and that the complex-valued fast fading random variables characterizing the propagation channels between different pairs of \gls{AP}-\gls{MS} connections are independent, the \gls{DL} and \gls{UL} achievable rates (measured in bits per second per Hertz) for \gls{MS} $k$ can be obtained as stated in the following theorems:
\begin{theorem}[Downlink achievable rate] An achievable rate of \gls{MS} k using the analog precoders $\bs{W}_m^{RF}$, for all $m\in\{1,\ldots,M\}$, and the \gls{ZF} baseband precoder $\bs{W}_d^{BB} = \hat{\bs{G}}^*\left(\hat{\bs{G}}^T \hat{\bs{G}}^*\right)^{-1}$ is ${R_d}_k=\log_2\left(1+{\SINR_d}_k\right)$, with
\begin{equation}
   {\SINR_d}_k=\frac{\upsilon_k}{\sum_{k'=1}^K \upsilon_{k'} \varpi_{kk'} + \sigma_{\eta_{dk}}^2},
   \label{eq:SINRdk}
\end{equation}
where
\begin{equation}
   \sigma_{\eta_{dk}}^2=\sum_{m=1}^M {\sigma_q^2}_{dm} \tr\left(\bs{R}_{mk}^{RF}\right) + \sigma_d^2,
\end{equation}
and
\begin{equation}
   \varpi_{kk'}=\left[\diag\left(\mathbb{E}\left\{{\bs{W}_d^{BB}}^H \tilde{\bs{g}}_k^* \tilde{\bs{g}}_k^T\bs{W}_d^{BB}\right\}\right)\right]_{k'}.
\end{equation}
\label{theo:theorem_DL}
\end{theorem}

\begin{proof}
See Appendix \ref{app:Appendix_1}.
\end{proof}

\begin{theorem}[Uplink achievable rate] An achievable \gls{UL} rate for the $k$th \gls{MS} in the Cell-Free Massive \gls{MIMO} system with limited capacity fronthaul links and using \gls{ZF} \gls{MIMO} detection, for any $M$, $N$ and $K$, is given by ${R_u}_k=\log_2\left(1+{\SINR_u}_k\right)$, with
\begin{equation}
   {\SINR_u}_k=\frac{P_u \omega_k}{P_u \sum_{k'=1}^K \omega_{k'} \delta_{kk'} + \sigma_{\eta_{uk}}^2},
   \label{eq:SINRuk}
\end{equation}
where
\begin{equation}
   \delta_{kk'}=\left[\diag\left(\mathbb{E}\left\{\tilde{\bs{G}}^H{\bs{w}_{uk}^{BB}}^H \bs{w}_{uk}^{BB} \tilde{\bs{G}}\right\}\right)\right]_{k'}
\end{equation}
with $\bs{w}_{uk}^{BB}$ denoting the $k$th row of $\bs{W}_u^{BB}$, or, equivalently,
\begin{equation}
   \delta_{kk'}=\left[\diag\left(\mathbb{E}\left\{\bs{W}_u^{BB} \tilde{\bs{g}}_{k'} \tilde{\bs{g}}_{k'}^H{\bs{W}_u^{BB}}^H\right\}\right)\right]_k,
\end{equation}
and
\begin{equation}
   \sigma_{\eta_{uk}}^2=\sum_{m=1}^M\left({\sigma_q^2}_{um} +\sigma_u^2(N)\right){\nu_u}_{mk},
\end{equation}
with
\begin{equation}
   {\nu_u}_{mk}=\left[\diag\left(\mathbb{E}\left\{\bs{W}_{u\,m}^{BB} {\bs{W}_{u\,m}^{BB}}^H\right\}\right)\right]_k.
\end{equation}
\label{theo:theorem_UL}
\end{theorem}

\begin{proof}
See Appendix \ref{app:Appendix_2}.
\end{proof}

\section{Fronthaul capacity consumption}
\label{sec:fronthaul_capacity}

The \gls{DL} quantization process performed at the $m$th \gls{AP} can be expressed as
\begin{equation}
   \hat{\mathcal{Q}}_{dm}\left(\bs{W}_{d\,m}^{BB}\bs{\Upsilon}^{1/2}\bs{s}_d\right)=\bs{W}_{d\,m}^{BB}\bs{\Upsilon}^{1/2}\bs{s}_d + {\bs{q}_d}_m.
\end{equation}
From standard random coding arguments \cite{Cover06}, vector $\bs{s}_d$ can be safely assumed to be distributed as $\bs{s}_d\sim\mathcal{CN}(0,\bs{I}_K)$ and thus, the quantized vector $\hat{\mathcal{Q}}_{dm}\left(\bs{W}_{d\,m}^{BB}\bs{\Upsilon}^{1/2}\bs{s}_d\right)$ is distributed as $\hat{\mathcal{Q}}_{dm}\left(\bs{W}_{d\,m}^{BB}\bs{\Upsilon}^{1/2}\bs{s}_d\right)\sim\mathcal{CN}\left(\bs{0},\bs{W}_{d\,m}^{BB}\bs{\Upsilon} {\bs{W}_{d\,m}^{BB}}^H+{\sigma_q^2}_{dm}\bs{I}_{L_A}\right)$. Furthermore, as the differential entropy of a vector $\bs{x}\sim\mathcal{CN}(\bs{\omega},\bs{\Theta})$ is given by $\mathcal{H}(\bs{x})=\log\det(\pi e \bs{\Theta})$ \cite{Cover06}, the required average rate to transfer the quantized vector $\hat{\mathcal{Q}}_{dm}\left(\bs{W}_{d\,m}^{BB}\bs{\Upsilon}^{1/2}\bs{s}_d\right)$ on the corresponding \gls{DL} fronthaul link can be obtained as (in bps/Hz)
\begin{equation}
\begin{split}
   \hat{C}_{dm} &= \mathbb{E}\left\{I\left(\hat{\mathcal{Q}}_{dm}\left(\bs{W}_{d\,m}^{BB}\bs{\Upsilon}^{1/2}\bs{s}_d\right);\bs{W}_{d\,m}^{BB}\bs{\Upsilon}^{1/2}\bs{s}_d\right)\right\} \\
                &= \mathbb{E}\left\{\mathcal{H}\left(\hat{\mathcal{Q}}_{dm}\left(\bs{W}_{d\,m}^{BB}\bs{\Upsilon}^{1/2}\bs{s}_d\right)\right)\right\} \\
                &\quad-\mathbb{E}\left\{\mathcal{H}\left(\hat{\mathcal{Q}}_{dm}\left(\bs{W}_{d\,m}^{BB}\bs{\Upsilon}^{1/2}\bs{s}_d\right)\bigr|\bs{W}_{d\,m}^{BB}\bs{\Upsilon}^{1/2}\bs{s}_d\right) \right\} \\
                &= \mathbb{E}\left\{\log_2\det\left(\frac{1}{{\sigma_q^2}_{dm}} \bs{W}_{d\,m}^{BB} \bs{\Upsilon} {\bs{W}_{d\,m}^{BB}}^H + \bs{I}_{L_A}\right)\right\},
\end{split}
\end{equation}
where $I(\hat{\bs{x}};\bs{x})$ is used to denote the mutual information between vectors $\hat{\bs{x}}$ and $\bs{x}$, and $\mathcal{H}(\hat{\bs{x}}|\bs{x})$ is the differential entropy of $\hat{\bs{x}}$ conditioned on $\bs{x}$. Since the determinant is a log-concave function on the set of positive semidefinite matrices, it follows from Jensen's inequality that
\begin{equation}
\begin{split}
   \hat{C}_{dm} &\leq \log_2\det\left(\frac{1}{{\sigma_q^2}_{dm}} \mathbb{E}\left\{\bs{W}_{d\,m}^{BB} \bs{\Upsilon} {\bs{W}_{d\,m}^{BB}}^H\right\} + \bs{I}_{L_A}\right) \\
                &= \log_2\det\left(\frac{1}{{\sigma_q^2}_{dm}} \sum_{k=1}^K \upsilon_k \bs{R}_{mk}^{BB} + \bs{I}_{L_A}\right),
\end{split}
\end{equation}
where $\bs{R}_{mk}^{BB} = \mathbb{E}\left\{\bs{w}_{mk}^{BB} {\bs{w}_{mk}^{BB}}^H\right\}$.

Analogously, the \gls{UL} quantization process performed at the $m$th \gls{AP} is given by $\hat{\mathcal{Q}}_{um}\left({\bs{r}_u}_m\right)={\bs{r}_u}_m + {\bs{q}_u}_m$. Thus, using arguments similar to those used in the \gls{DL} case, the required average rate to transfer the quantized vector $\hat{\mathcal{Q}}_{um}\left({\bs{r}_u}_m\right)$ on the corresponding \gls{UL} fronthaul link can be upper bounded as (in bps/Hz)
\begin{equation}
\begin{split}
   &\hat{C}_{um} = \mathbb{E}\left\{I\left(\hat{\mathcal{Q}}_{um}\left({\bs{r}_u}_m\right);{\bs{r}_u}_m\right)\right\} \\
   &\ = \mathbb{E}\left\{\mathcal{H}\left(\hat{\mathcal{Q}}_{um}\left({\bs{r}_u}_m\right)\right)\right\}-\mathbb{E}\left\{\mathcal{H}\left(\hat{\mathcal{Q}}_{um}\left({\bs{r}_u}_m\right)\bigr|{\bs{r}_u}_m\right) \right\} \\
   &\ \leq \log_2\det\left(\frac{P_u}{{\sigma_q^2}_{um}} \sum_{k=1}^K \omega_k \bs{R}_{mk}^{RF} + \left(\frac{\sigma_u^2(N)}{{\sigma_q^2}_{um}} +1\right)\bs{I}_{L_A}\right).
\end{split}
\end{equation}

\section{Pilot assignment}
\label{sec:pilot_assignment}

To warrant an appropriate system performance, the \gls{RRM} unit must efficiently manage both the pilot assignment and the \gls{UL} and \gls{DL} power control. As the pilots are not power controlled, pilot assignment and power control can be conducted independently. Since the length of the pilot sequences is limited to $\tau_p$, there only exist $\tau_p$ orthogonal pilot sequences. In a network with $K\leq \tau_p$ \glspl{MS}, an optimal pilot assignment strategy simply allocates $K$ orthogonal pilots to the $K$ \glspl{MS}. The real pilot assignment problem arises when $K > \tau_p$. In this case, fully orthogonal pilot assignment is no longer possible and hence, other pilot assignment strategies must be devised.

On the one hand, designing an optimal pilot assignment strategy aiming at maximizing the minimum rate allocated to the active \glspl{MS} in the network is a very difficult combinatorial problem, computationally unmanageable in most network setups of practical interest \cite{Ngo17}. On the other hand, using straightforward strategies such as, for instance, the pure \gls{RPA} scheme \cite{Ahmadi16}, where each \gls{MS} is randomly assigned one pilot sequence out of the set of $\tau_p$ orthogonal pilot sequences, or the \gls{BRPA} scheme, where each \gls{MS} is allocated a pilot sequence that is sequentially and cyclically selected from the ordered set of available orthogonal pilots, provides poor performance results. In order to avoid the computational complexity of the optimal strategies while improving the performance of the baseline \gls{RPA} or \gls{BRPA} approaches, a suboptimal solution is proposed in this paper that is based on the idea of \emph{clustering by dissimilarity}. This suboptimal approach, that will be termed as the \gls{DCPA} strategy, is motivated by the following key observation:

\textbf{Key observation}: \emph{In those scenarios where $K > \tau_p$, cell-free communication is severely impaired whenever \glspl{MS} showing very similar large-scale propagation patterns to the set of \glspl{AP} (that is, \glspl{MS} typically located nearby) are allocated the same pilot sequence. In this case, the inter-\gls{MS} interference leads to very poor channel estimates at all \glspl{AP} and, eventually, to low \glspl{SINR}.}

The clustering algorithm proposed in this work basically ensures that pilot sequences are only reused by \glspl{MS} showing \emph{dissimilar} large-scale propagation patterns to the \glspl{AP} (that is, \glspl{MS} typically located sufficiently apart). Two key aspects regarding the clustering operation are thus, on the one hand, to decide which should be the large-scale propagation pattern that ought to be used to represent a given \gls{MS} and, on the other hand, to decide what metric should be used to measure \emph{similarity} among the large-scale propagation patterns characterizing different \glspl{MS}. To this end, and resting upon the premise that the \gls{CPU} has perfect knowledge of the large-scale gains, let $\bs{\xi}_k=\left[\xi_{1k}\,\ldots\,\xi_{Mk}\right]^T$ denote the $M\times 1$ vector containing the average Frobenius norms of the equivalent channels linking the $k$th \gls{MS} to all $M$ \glspl{AP} in the cell-free network. Vector $\bs{\xi}_k$ can be considered as an effective \emph{fingerprint} characterizing the location of \gls{MS} $k$. Now, although no single definition of a similarity measure exists, the so-called \emph{cosine similarity} measure is one of the most commonly used similarity metrics when dealing with real-valued vectors. Hence, as the \emph{fingerprint} vectors characterizing the different \glspl{MS} are non-negative real-valued, the cosine similarity measure between two \emph{fingerprint} vectors $\bs{\xi}_k$ and $\bs{\xi}_{k'}$, defined as
\begin{equation}
   f_D\left(\bs{\xi}_k, \bs{\xi}_{k'}\right)=\frac{\bs{\xi}_k^T\bs{\xi}_{k'}}{\|\bs{\xi}_k\|_2 \|\bs{\xi}_{k'} \|_2},
\end{equation}
will be used as a proper similarity metric in our work. The resulting similarity values range from 0, meaning orthogonality (perfect dissimilarity), to 1, meaning exact match (perfect similarity).

The proposed \gls{DCPA} algorithm proceeds as follows. In a first step, it calculates the fingerprint of an imaginary \gls{MS} centroid, defined as
\begin{equation}
   \bs{\xi}_C=\frac{1}{K}\sum_{k=1}^K \bs{\xi}_k.
\end{equation}
Then, it moves onward to the calculation of the cosine similarity measures among the fingerprint vectors characterizing the $K$ \glspl{MS} in the network and the fingerprint of the centroid, that is, the algorithm proceeds to the calculation of $f_D\left(\bs{\xi}_k, \bs{\xi}_C\right)$, for all $k\in\{1,\ldots,K\}$. The \glspl{MS} are then sorted in descending order of similarity with the centroid, that is, the algorithm obtains the ordered set of subindices $\mathcal{O}=\left\{o_1, o_2, \ldots, o_K\right\}$, such that $f_D\left(\bs{\xi}_{o_1}, \bs{\xi}_C\right) \leq f_D\left(\bs{\xi}_{o_2}, \bs{\xi}_C\right)\leq \cdots \leq f_D\left(\bs{\xi}_{o_K}, \bs{\xi}_C\right)$. Once the \glspl{MS} have been sorted, the algorithm constructs $\tau_p$ clusters of \glspl{MS}, namely $\mathcal{K}_1, \ldots, \mathcal{K}_{\tau_p}$, with
\begin{equation}
\begin{split}
   \mathcal{K}_t=&\mathcal{O}\left(t:\tau_p:K\right) \\
                =&\left\{o_t, o_{t+\tau_p}, o_{t+2 \tau_p},\ldots\right\},\ \forall t\in\{1,\ldots,\tau_p\},
\end{split}
\end{equation}
and all \glspl{MS} in cluster $\mathcal{K}_t$, which are located far from each other, are allocated the same pilot code $\bs{\varphi}_t$. Note that the application of this algorithm ensures that, as far as it is possible, two \glspl{MS} having similar large-scale propagation fingerprints are allocated different pilot codes and, thus, they do not interfere to each other during the \gls{UL} channel estimation process. In other words, it aims at minimizing the residual interuser interference terms in both \eqref{eq:ydk} and \eqref{eq:yuk}.

\section{Max-min power allocation and optimal quantization}
\label{sec:power allocation_quantization}

\subsection{Downlink power control and quantization}

In line with previous research works on cell-free architectures \cite{Ngo15,Ngo17,Nayebi17,Bashar18b}, our aim in this subsection is to find the power control coefficients $\upsilon_k$, for all $k\in\{1,\ldots,K\}$, and the quantization noise variances ${\sigma_q^2}_{dm}$, for all $m\in\{1,\ldots,M\}$, that maximize the minimum of the achievable \gls{DL} rates of all \glspl{MS} while satisfying the average transmit power and \gls{DL} fronthaul capacity constraints at each \gls{AP}. Mathematically, this optimization problem can be formulated as
\begin{equation}
\begin{split}
   &\max_{\substack{\bs{\Upsilon} \succeq 0 \\ {\bs{\sigma}_q}_d \succeq 0}}\ \min_{k\in\{1,\ldots,K\}} \frac{\upsilon_k}{\sum_{k'=1}^K \upsilon_{k'} \varpi_{kk'} + \sigma_{\eta_{dk}}^2} \\
   &\textrm{s.t. } \sum_{k=1}^K \upsilon_k \theta_{mk}^{BB/RF} \leq \overline{P}_m - {\sigma_q^2}_{dm} L_A N,\,\forall\,m, \\
   &\phantom{\textrm{s.t. }} \log_2\det\left(\sum_{k=1}^K \frac{\upsilon_k}{{\sigma_q^2}_{dm}} \bs{R}_{mk}^{BB} + \bs{I}_{L_A}\right) \leq {C_F}_d,\,\forall\,m,
\end{split}
\label{eq:opt_problem_DL}
\end{equation}
where we have used the definition ${\bs{\sigma}_q}_d=[{\sigma_q}_{d1}\,\ldots\,{\sigma_q}_{dM}]^T$.

%\subsubsection{\Gls{GOPA}}
Optimization problem \eqref{eq:opt_problem_DL} is characterized by continuous objective and constraint functions of interdependent block variables, namely, $\bs{\Upsilon}$ and ${\bs{\sigma}_q}_d$.  A widely used approach for solving optimization problems of this class is the so-called \gls{BCD} method. This is an iterative optimization approach that, at each iteration and in a cyclic order, optimizes one of the blocks while the remaining variables are held fixed \cite{Tseng01,Beck13}. Convergence of the \gls{BCD} method is ensured whenever each of the subproblems to be optimized in each iteration can be exactly solved to its unique optimal solution.

The first important fact to note is that, given a power allocation matrix $\bs{\Upsilon}^{(i-1)}$ obtained at the $(i-1)$th iteration, and as the achievable user rates monotonically increase with the capacity of the fronthaul links between the \glspl{AP} and the \gls{CPU}, the optimal solution for the acceptable fronthaul quantization noise in the $i$th iteration is achieved when the fronthaul capacity constraints are satisfied with equality, that is, when
\begin{equation}
   \det\left(\sum_{k=1}^K \frac{\upsilon_k^{(i-1)}}{{\sigma_q^2}_{dm}^{(i)}} \bs{R}_{mk}^{BB} + \bs{I}_{L_A}\right) = 2^{{C_F}_d},\,\forall\,m.
   \label{eq:trasc_funct}
\end{equation}
Note that ${\sigma_q^2}_{dm}^{(i)}$ cannot be expressed in a closed-form algebraic expression as it only admits a solution in the form of a transcendental function
\begin{equation}
   {\sigma_q^2}_{dm}^{(i)}=F_d\left(\bs{\Upsilon}^{(i-1)}, \left\{\bs{R}_{mk}^{BB}\right\}_{k=1}^K, {C_F}_d\right)
\end{equation}
that can be numerically solved by applying mathematical software tools to \eqref{eq:trasc_funct}.

Once the optimal block of variables ${\bs{\sigma}_q}_d^{(i)}$ have been obtained, the optimization problem in \eqref{eq:opt_problem_DL} can be rewritten in terms of the power allocation matrix $\bs{\Upsilon}^{(i)}$ as
\begin{equation}
\begin{split}
   &\max_{\bs{\Upsilon}^{(i)} \succeq 0}\ \min_{k\in\{1,\ldots,K\}} \frac{\upsilon_k^{(i)}}{\displaystyle{\small{\sum_{k'=1}^K} \upsilon_{k'}^{(i)} \gamma_{kk'} + \small{\sum_{m=1}^M} {\sigma_q^2}_{dm}^{(i)} \tr\left(\bs{R}_{mk}^{RF}\right) + \sigma_d^2}} \\
   &\textrm{s.t. } \sum_{k=1}^K \upsilon_k^{(i)} \theta_{mk}^{BB/RF} \leq \overline{P}_m - N L_A {\sigma_q^2}_{dm}^{(i)},\,\forall\,m.
\end{split}
\label{eq:opt_problem_DL_GOPA}
\end{equation}
Note that this is a convergent quasi-linear optimization problem that can be solved using conventional standard convex optimization methods \cite{Ngo17,Nayebi17}.

\subsection{Uplink power control and quantization}

In this subsection we aim at finding the power control coefficients $\omega_k$, for all $k\in\{1,\ldots,K\}$, and quantization noise variances ${\sigma_q^2}_{um}$, for all $m\in\{1,\ldots,M\}$, that maximize the minimum of the achievable ulink rates of all \glspl{MS} while satisfying the power control coefficient constraints at each \gls{MS} and the \gls{UL} fronthaul capacity constraints at each \gls{AP}. This optimization problem can be formulated as
\begin{equation}
\begin{split}
   &\max_{\substack{\bs{\omega} \succeq 0 \\ {\bs{\sigma}_q}_u \succeq 0}}\ \min_{k\in\{1,\ldots,K\}} \frac{P_u \omega_k}{P_u \sum_{k'=1}^K \omega_{k'} \delta_{kk'} + \sigma_{\eta_{uk}}^2} \\
   &\textrm{s.t. } 0 \leq \omega_k \leq 1,\,\forall\,k, \\
   &\phantom{\textrm{s.t. }} \det\left(\frac{P_u}{{\sigma_q^2}_{um}} \sum_{k=1}^K \omega_k \bs{R}_{mk}^{RF} + \vartheta_m\bs{I}_{L_A}\right) \leq 2^{{C_F}_u},\,\forall\,m,
\end{split}
\label{eq:opt_problem_UL}
\end{equation}
where ${\bs{\sigma}_q}_u=[{\sigma_q}_{u1}\,\ldots\,{\sigma_q}_{uM}]^T$, and we have used the definition $\vartheta_m= 1+\sigma_u^2(N)/{\sigma_q^2}_{um}$. As for the \gls{DL} case, problem \eqref{eq:opt_problem_UL} admits the use of the \gls{BCD} method where, in each iteration, the nonconvex transcendental function ${\sigma_q^2}_{um}=F_u\left(\bs{\Omega}, \left\{\bs{R}_{mk}^{RF}\right\}_{k=1}^K, P_u, {C_F}_u\right)$ is approximated by a constant calculated using the power allocation vector obtained in the previous iteration of the algorithm. That is, in the $i$th iteration of the \gls{UL} optimal power allocation approach, the algorithm solves the optimization problem
\begin{equation}
\begin{split}
   &\max_{\bs{\Omega}^{(i)} \succeq 0}\ \min_{k\in\{1,\ldots,K\}} \frac{P_u \omega_k^{(i)}}{P_u \sum_{k'=1}^K \omega_{k'}^{(i)} \delta_{kk'} + {\sigma_{\eta_{uk}}^{\,2(i)}}}, \\
   &\textrm{s.t. } 0 \leq \omega_k \leq 1,\,\forall\,k,
\end{split}
\label{eq:opt_problem_UL_SOPA}
\end{equation}
where ${\sigma_q^2}_{um}^{(i)}=F_u\left(\bs{\Omega}^{(i-1)}, \left\{\bs{R}_{mk}^{RF}\right\}_{k=1}^K, P_u, {C_F}_u\right)$. Note that, again, this is a convergent quasi-linear optimization problem that can be solved using conventional convex optimization methods \cite{Ngo17,Nayebi17}.
%In order to obtain an initial guess for the solution of this problem, we propose to solve the quasi-linear optimization problem
%\begin{equation}
%\begin{split}
%   &\max_{\bs{\Omega}^{(0)} \succeq 0}\ \min_{k\in\{1,\ldots,K\}} \frac{P_u \omega_k^{(0)}}{P_u \sum_{k'=1}^K \omega_{k'}^{(0)} \delta_{kk'} + \sigma_u^2(N)\sum_{m=1}^M {\nu_u}_{mk}}, \\
%   &\textrm{s.t. } 0 \leq \omega_k \leq 1,\,\forall\,k,
%\end{split}
%\label{eq:opt_problem_UL_inf_front}
%\end{equation}
%That is, the power allocation vector used to obtain the variance of the quantization noise in the first iteration of the algorithm is obtained assuming the availability of infinite capacity fronthauls.

\section{Numerical results}
\label{sec:numerical_results}

In this section, simulation results are obtained in order to quantitatively study the performance of the proposed cell-free \gls{mmWave} massive \gls{MIMO} network with constrained-capacity fronthaul links. In particular, we demonstrate the impact of using different pilot allocation strategies, the effects of modifying the capacity of the fronthaul links and the \gls{RF} infrastructure at the \glspl{AP}, and the repercussion of changing the density of \glspl{AP} per area unit. For simplicity of exposition, and without loss of essential generality, a cell-free scenario is considered where the $M$ \glspl{AP} and $K$ \glspl{MS} are uniformly distributed at random within a square coverage area of size $D \times D$~$m^2$. As described in subsection \ref{subsec:Channel_model}, a modified version of the discrete-time narrowband clustered channel model proposed by Akdeniz \emph{et al.} in \cite{Akdeniz14} is used in the performance evaluation. The parameters necessary to implement this channel model can be found in \cite[Table I]{Akdeniz14}. Furthermore, similar to what was done by Ngo \emph{et al.} in \cite{Ngo17}, a shadow fading spatial correlation model with two components is also considered (see \cite[eqs. (54) and (55)]{Ngo17}) where the decorrelation distance is set to $d_{decorr}=50$~m and the parameter $\delta$ is set to 0.5. Default parameters used to set-up the simulation scenarios under evaluation in the following subsections are summarized in Table \ref{tab:parameters}.

\begin{table}[!t]
\renewcommand{\arraystretch}{1.2}
\caption{\small Summary of default simulation parameters}
\label{tab:parameters}
\centering
\begin{tabular}{l|c}
\hline
\bfseries Parameters & \bfseries Value\\
\hline
\footnotesize {Carrier frequency: $f_0$} & \footnotesize {28 GHz}\\
\footnotesize {Bandwidth: $B$} & \footnotesize {20 MHz}\\
\footnotesize {Side of the square coverage area: $D$} & \footnotesize {200 m}\\
\footnotesize {AP antenna height: $h_{AP}$} & \footnotesize {15 m}\\
\footnotesize {MS antenna height: $h_{MS}$} & \footnotesize {1.65 m}\\
\footnotesize {Noise figure at the MS: ${NF}_{MS}$} & \footnotesize {9 dB} \\
\footnotesize {Noise figure of the LNA at the AP: ${NF}_{LNA}$} & \footnotesize {1.6 dB} \\
\footnotesize {Gain of the LNA at the AP: $G_{LNA}$} & \footnotesize {22 dB} \\
\footnotesize {Attenuation of the phase splitters at the AP: $L_{PS}$} & \footnotesize {3 dB}\\
\footnotesize {Attenuation of the power combiner at the AP: ${L_{PC}}_{in}$} & \footnotesize {3 dB}\\
\footnotesize {Noise figure of the RF chain at the AP: $NF_{RF}$} & \footnotesize {7 dB}\\
\footnotesize {Available average power at the AP: $\overline{P}_m$} & \footnotesize {200 mW}\\
\footnotesize {Available average power at the MS: $P_u=P_p$} & \footnotesize {100 mW}\\
\footnotesize {Coherence interval length: $\tau_c$} & \footnotesize {200 samples}\\
\footnotesize {Training phase length: $\tau_p$} & \footnotesize {15 samples}\\
\hline
\end{tabular}
\end{table}

\begin{figure}[t!]

    \centering
    \begin{subfigure}[t]{0.24\textwidth}
        \centering
        \includegraphics[height=7cm]{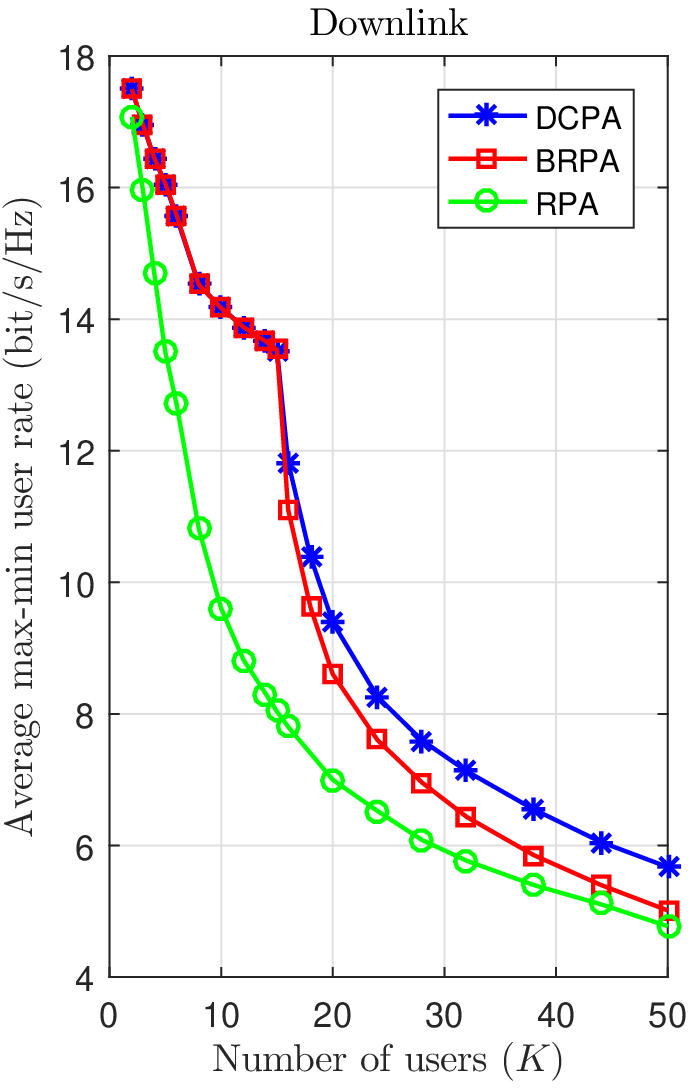}
%        \caption{\small Downlink}
        \label{fig:DL_var_Pilot_Alloc_min_rate}
    \end{subfigure}
    \begin{subfigure}[t]{0.24\textwidth}
        \centering
        \includegraphics[height=7cm]{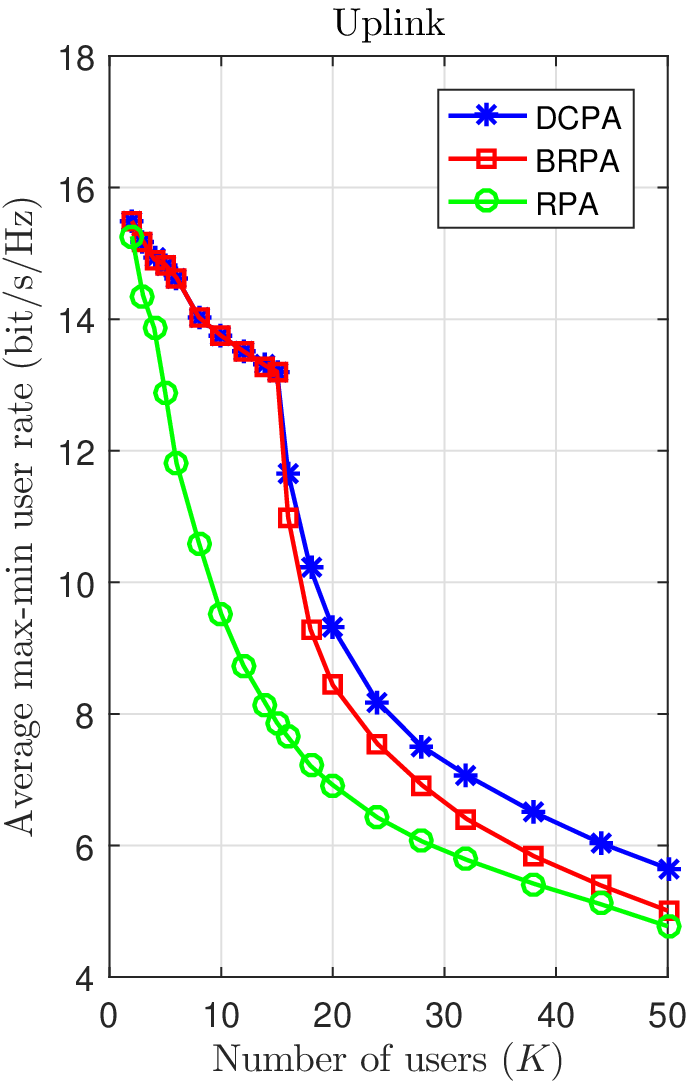}
%        \caption{\small Uplink}
        \label{fig:UL_var_Pilot_Alloc_min_rate}
    \end{subfigure}

    \caption{\small Average max-min rate per user \emph{versus} the number of active \glspl{MS} for different pilot allocation strategies ($N=64$ antennas, $L=8$ \gls{RF} chains, ${C_F}_d={C_F}_u=64$ bit/s/Hz).}
    \label{fig:var_Pilot_Alloc_min_rate}
\end{figure}

\subsection{Impact of the pilot allocation process}

Our aim in this subsection is to benchmark the performance of the proposed large-scale \gls{CSI}-aware \gls{DCPA} strategy against both the pure \gls{RPA} and the \gls{BRPA} schemes. Accordingly, the average max-min rate per user \emph{versus} the number of active \glspl{MS} is presented in Fig. \ref{fig:var_Pilot_Alloc_min_rate} for each of these pilot allocation strategies and for both the \gls{DL} and the \gls{UL}. All results have been obtained assuming the default system parameters described in Table \ref{tab:parameters}, the use of $L=8$ \gls{RF} chains fully connected to uniform linear antenna arrays with $N=64$ antenna elements, and fronthaul links with a capacity of ${C_F}_d={C_F}_u=64$ bit/s/Hz. The first important result to note from Fig. \ref{fig:var_Pilot_Alloc_min_rate} is that the pure \gls{RPA} scheme is clearly outperformed by both the \gls{BRPA} and the \gls{DCPA} strategies irrespective of the of active \glspl{MS} in the network. In fact, the \gls{RPA} scheme cannot guarantee neither the absence of pilot reuse, even for those cases in which $K\leq\tau_p$ (in this setup, $\tau_p=15$ time/frequency samples), nor the possibility of having pilots that are allocated to a high number of \glspl{MS} and/or to \glspl{MS} exhibiting very similar large-scale propagation patterns to the \glspl{AP}. Therefore, the higher the number of active \glspl{MS}, the higher the probability of having one or more users suffering from high levels of pilot contamination, with the consequent reduction of the achievable max-min user rate. If we turn our attention to results provided by the \gls{BRPA} and \gls{DCPA} strategies, two disjoint operation regions can be distinguished. In the first one, comprising the scenarios in which $K\leq\tau_p$, both approaches allocate orthogonal pilots to the users (absence of pilot contamination) and thus naturally provide the same performance. In the second one, however, comprising the scenarios in which $K>\tau_p$, pilots have to be reused and, as a consequence, pilot contamination appears (note the rather abrupt performance drop when going from $K\leq\tau_p$ to $K>\tau_p$). In these scenarios, based on a smart exploitation of the available large-scale \gls{CSI}, the proposed \gls{DCPA} approach reduces the amount of pilot contamination experienced by the worst users in the network and it clearly improves the achievable max-min user rates provided by the channel-unaware \gls{BRPA} scheme.

Another result that is worth emphasizing, since it will repeatedly appear in the following subsections, is that, although in scenarios with high-capacity fronthaul links the achievable max-min \gls{DL} user rate is higher than that provided in the \gls{UL}, as the number of active users in the network increases, the performance obtained in both the \gls{DL} and the \gls{UL} tend to become increasingly similar. This behavior can be easily deduced from the analysis of the \gls{SINR} expressions in \eqref{eq:SINRdk} and \eqref{eq:SINRuk}. As the number of active \glspl{MS} in the cell-free network increases, provided that it is greater than $\tau_p$, the term in the denominator corresponding to the residual interuser interference due to pilot contamination becomes increasingly dominant in comparison to the quantification and thermal noise terms, eventually reaching the point where they can be considered virtually negligible. Under these conditions, and since the pre-coding filters used on both links are identical, the \gls{DL} and the \gls{UL} experience similar \gls{SINR} values and, therefore, tend to provide the same achievable max-min rate per user, except for small differences that can be attributed to, on the one and, the dissimilar amount of quantified information that has to be conveyed through the corresponding fronthaul links and, on the other hand, disparities among the thermal noise powers experienced at both the \glspl{AP} and the \glspl{MS}.

\begin{figure}[t!]

    \centering
    \begin{subfigure}[t]{0.24\textwidth}
        \centering
        \includegraphics[height=7cm]{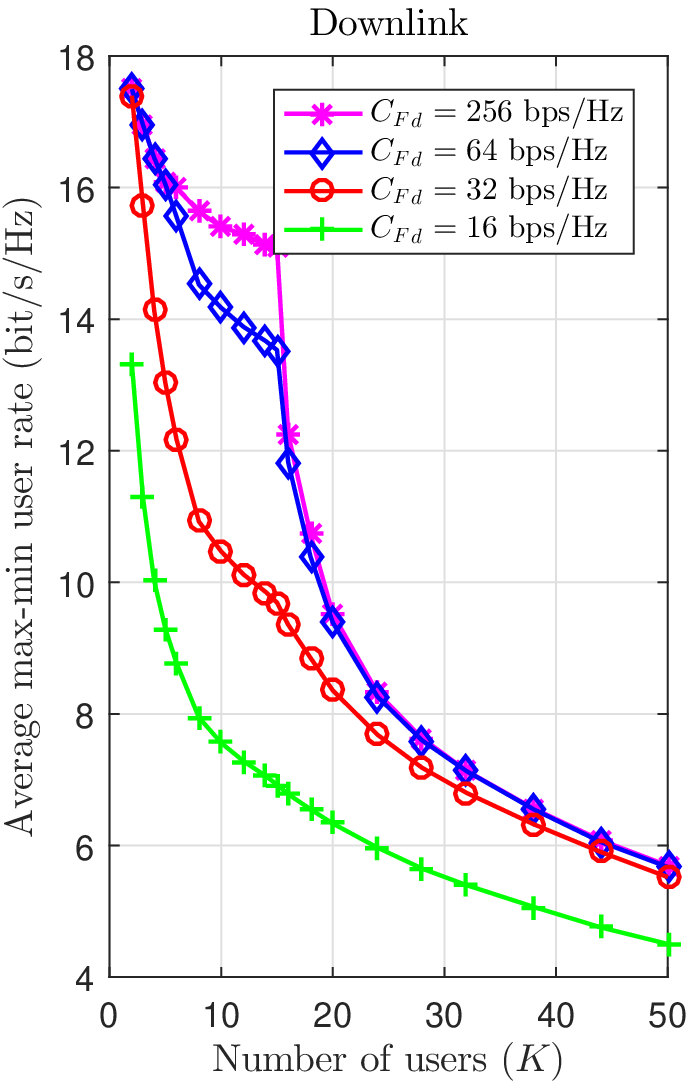}
%        \caption{\small Downlink}
        \label{fig:DL_varCFd_min_rate}
    \end{subfigure}
    \begin{subfigure}[t]{0.24\textwidth}
        \centering
        \includegraphics[height=7cm]{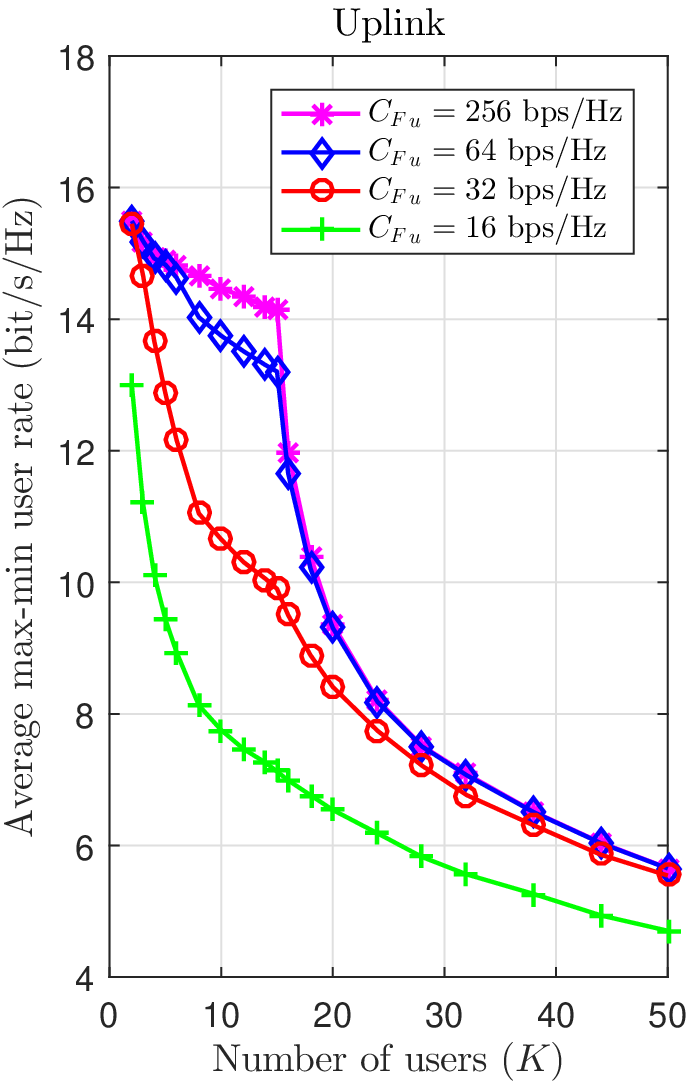}
%        \caption{\small Uplink}
        \label{fig:DL_varCFu_min_rate}
    \end{subfigure}

    \caption{\small Average max-min rate per user \emph{versus} the number of active \glspl{MS} for different values of the fronthaul capacities ($N=64$ antennas, $L=8$ \gls{RF} chains, \gls{DCPA}).}
    \label{fig:varCF_min_rate}
\end{figure}

\subsection{Modifying the capacity of the fronthaul links and the \gls{RF} infrastructure at the \glspl{AP}}

The max-min achievable rate per user is plotted in Fig. \ref{fig:varCF_min_rate} against the number of active \glspl{MS} in the network, assuming the use of fronthaul links with different constraining capacities equal to 16, 32, 64 and 256 bit/s/Hz (for the network setups under consideration, using fronthaul links with a capacity of 256 bit/s/Hz is virtually equivalent to using infinite-capacity fronthauls). As expected, results show that increasing the fronthaul capacity is always beneficial if the main aim is to increase the achievable max-min user rate. Nevertheless, it is worth stressing that, keeping all the other parameters constant, the marginal increment of performance produced by each new increment of the fronthaul capacity suffers from the law of diminishing returns, especially for network setups with a high number of active \glspl{MS}. That is, although the performance increase produced by doubling the fronthaul capacity from 16 bit/s/Hz to 32 bit/s/Hz, or even from 32 bit/s/Hz to 64 bit/s/Hz, can be justifiable, increasing the fronthaul capacity beyond 64 bit/s/Hz does not seem to be reasonable from the point of view of increasing the achievable performance of the system under the considered network setups. As observed in the previous subsection, in cell-free \gls{mmWave} massive \gls{MIMO} networks using high-capacity fronthaul links, the achievable max-min \gls{DL} user rate is always slightly higher than that achieved in the \gls{UL} irrespective of the number of active \glspl{MS}. In scenarios with low-capacity fronthaul links and a large number of active \glspl{MS}, however, the quantization noise experienced in the \gls{DL} is higher than its \gls{UL} counterpart and thus, the achievable per-user rate in the \gls{UL} is slightly higher that than supplied in the \gls{DL}.

\begin{figure}[t!]

    \centering
    \begin{subfigure}[t]{0.24\textwidth}
        \centering
        \includegraphics[height=7cm]{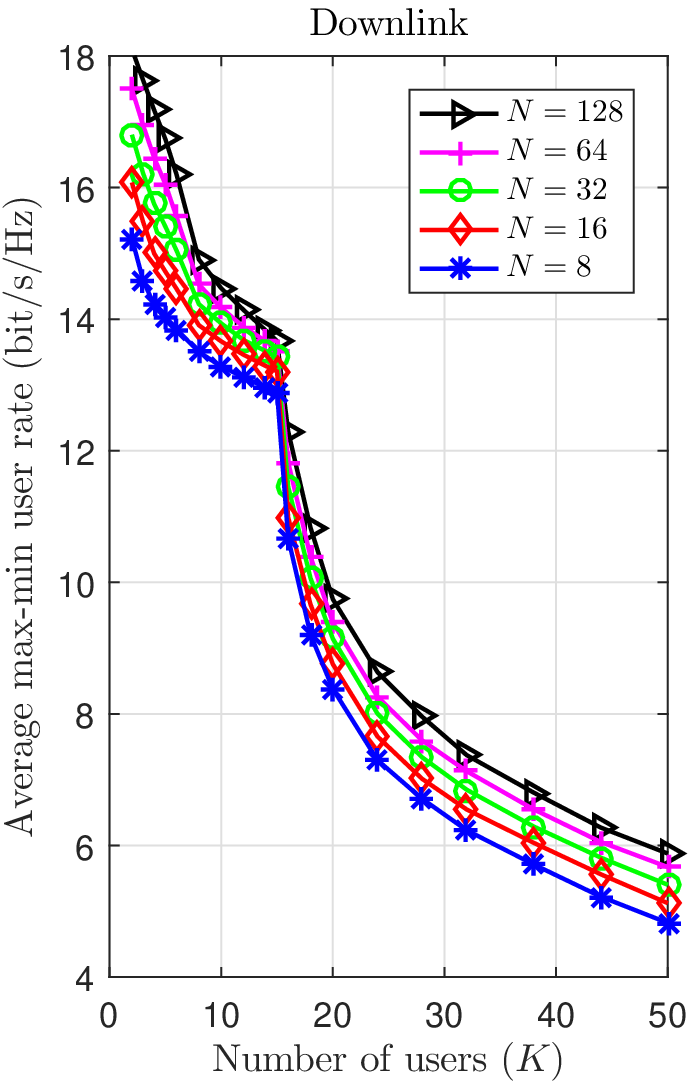}
%        \caption{\small Downlink}
        \label{fig:DL_varN_min_rate}
    \end{subfigure}
    \begin{subfigure}[t]{0.24\textwidth}
        \centering
        \includegraphics[height=7cm]{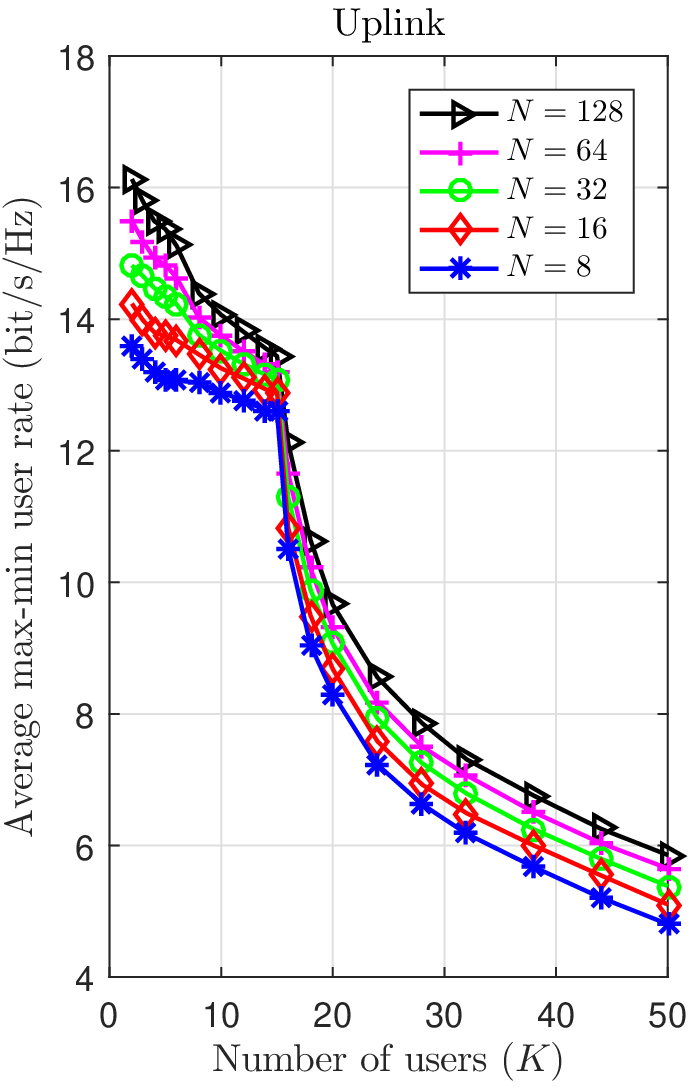}
%        \caption{\small Uplink}
        \label{fig:DL_varN_min_rate}
    \end{subfigure}

    \caption{\small Average max-min rate per user \emph{versus} the number of active \glspl{MS} for different values of the number of antennas at the \glspl{AP} ($L=8$ \gls{RF} chains, ${C_F}_d={C_F}_u=64$ bit/s/Hz, \gls{DCPA}).}
    \label{fig:varN_min_rate}
\end{figure}

\begin{figure}[t!]

    \centering
    \begin{subfigure}[t]{0.24\textwidth}
        \centering
        \includegraphics[height=7cm]{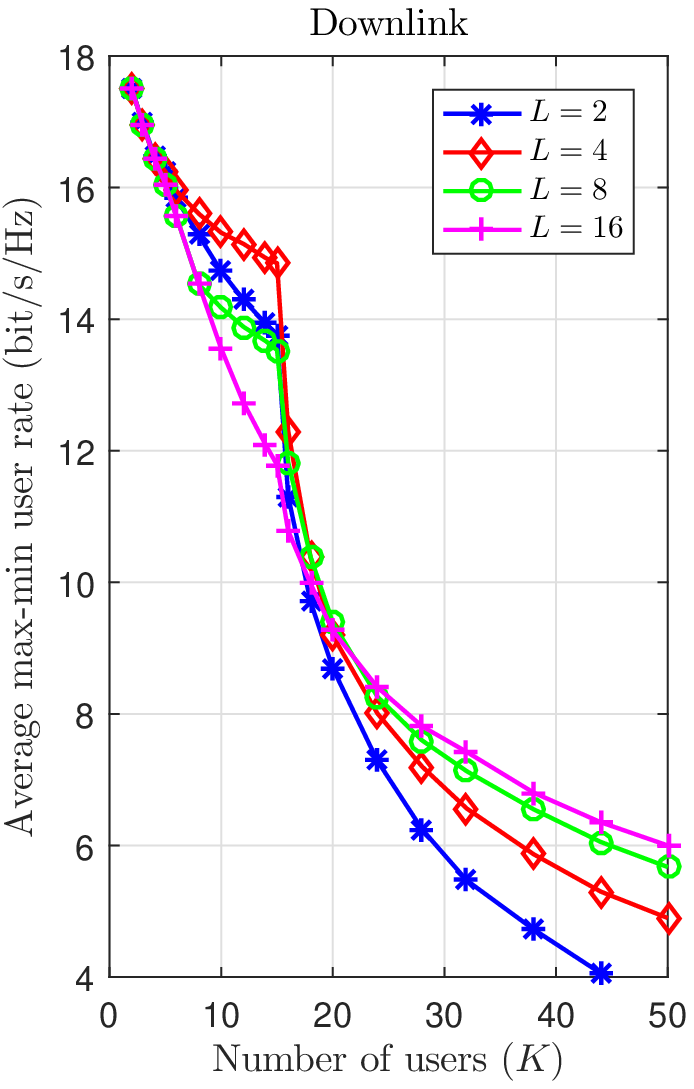}
%        \caption{\small Downlink}
        \label{fig:DL_varL_min_rate}
    \end{subfigure}
    \begin{subfigure}[t]{0.24\textwidth}
        \centering
        \includegraphics[height=7cm]{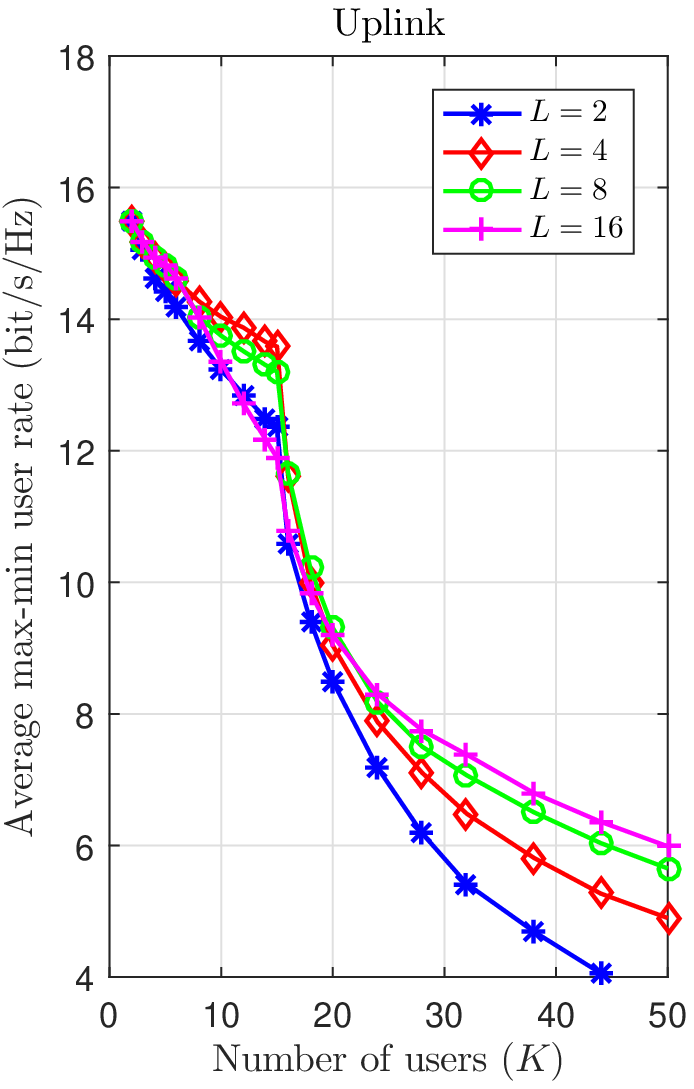}
%        \caption{\small Uplink}
        \label{fig:DL_varL_min_rate}
    \end{subfigure}

    \caption{\small Average max-min rate per user \emph{versus} the number of active \glspl{MS} for different values of the number of \gls{RF} chains at the \glspl{AP} ($N=64$ antennas, ${C_F}_d={C_F}_u=64$ bit/s/Hz, \gls{DCPA}).}
    \label{fig:varL_min_rate}
\end{figure}

To understand how the \gls{RF} infrastructure used at the \glspl{AP} influences the performance of the proposed cell-free \gls{mmWave} massive \gls{MIMO} system under constrained-capacity fronthaul links, Figs. \ref{fig:varN_min_rate} and \ref{fig:varL_min_rate} show the achievable max-min user rate against the number of active \glspl{MS} assuming the use of uniform linear antenna arrays with different number of elements and fully-connected analog \gls{RF} precoders with different number of \gls{RF} chains, respectively. In particular, results presented in Fig. \ref{fig:varN_min_rate} have been obtained assuming the use of an analog precoder with $L=8$ \gls{RF} chains fully-connected to a linear uniform antenna array with $N=8$, 16, 32, 64 or 128 antenna elements, whereas results presented in Fig. \ref{fig:varL_min_rate} have been obtained assuming the use of $L=2$, 4, 8 or 16 \gls{RF} chains fully-connected to a linear uniform antenna array with $N=64$ antenna elements. The first conclusion we may draw when looking at the results presented in Fig. \ref{fig:varN_min_rate} is that, irrespective of the number of active \glspl{MS} in the cell-free network, increasing the number of antenna elements at the \glspl{AP} in scenarios with high capacity fronthaul links (${C_F}_d={C_F}_u=64$ bit/s/Hz), although moderate and subject to the law of diminishing returns, always produces an increase in the achievable max-min user rate. As shown in Fig. \ref{fig:varL_min_rate}, in contrast, the impact produced by an increase in the number of \gls{RF} chains at the \glspl{AP} depends on the number of active \glspl{MS} in the network. In particular, when the number of active users is high, the interuser interference term due to pilot contamination (imperfect \gls{CSI}) dominates the factors in the denominator of the \gls{SINR} (i.e., makes the quantization and thermal noises negligible) and thus, increasing the number of \gls{RF} chains is always beneficial when trying to increase the achievable max-min user rate. When the number of active users in the network is low, however, the quantization noise, which is an increasing function of $L$, is not negligible anymore when compared to the interuser interference term (recall that this term is null when the number of active \glspl{MS} is less than or equal to $\tau_p$) and thus, increasing the number of \gls{RF} chains at the \glspl{AP} can be clearly disadvantageous.

\begin{figure}

    \centering
    \begin{subfigure}[t]{0.24\textwidth}
        \centering
        \includegraphics[height=7cm]{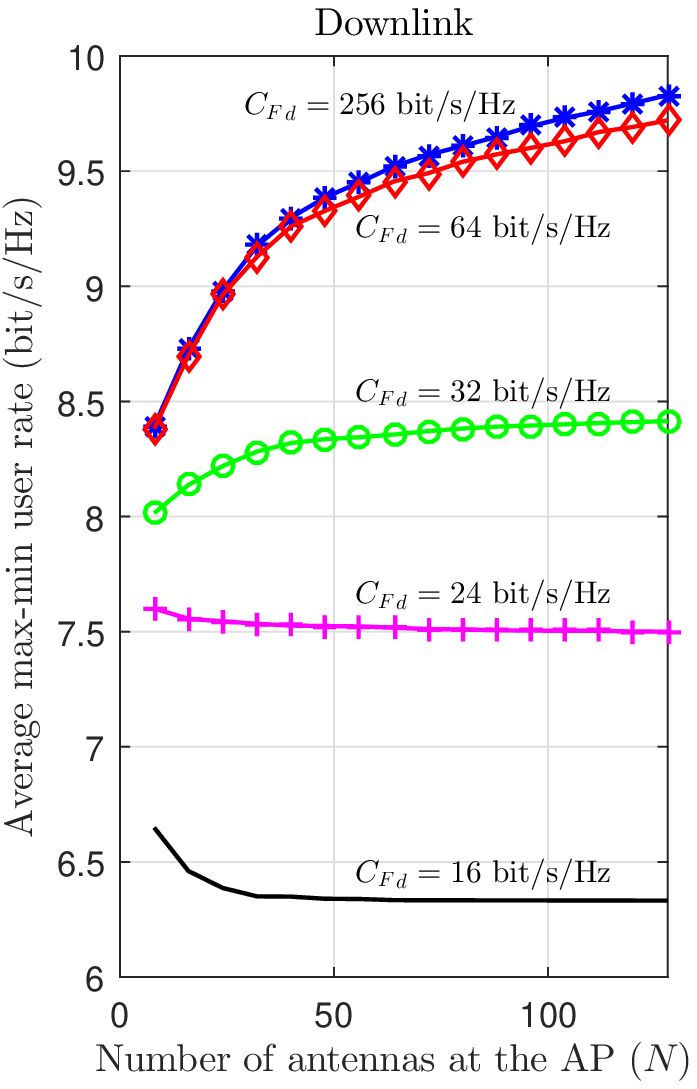}
%        \caption{\small Downlink}
        \label{fig:DL_FixK_varN_varCFd_min_rate}
    \end{subfigure}
    \begin{subfigure}[t]{0.24\textwidth}
        \centering
        \includegraphics[height=7cm]{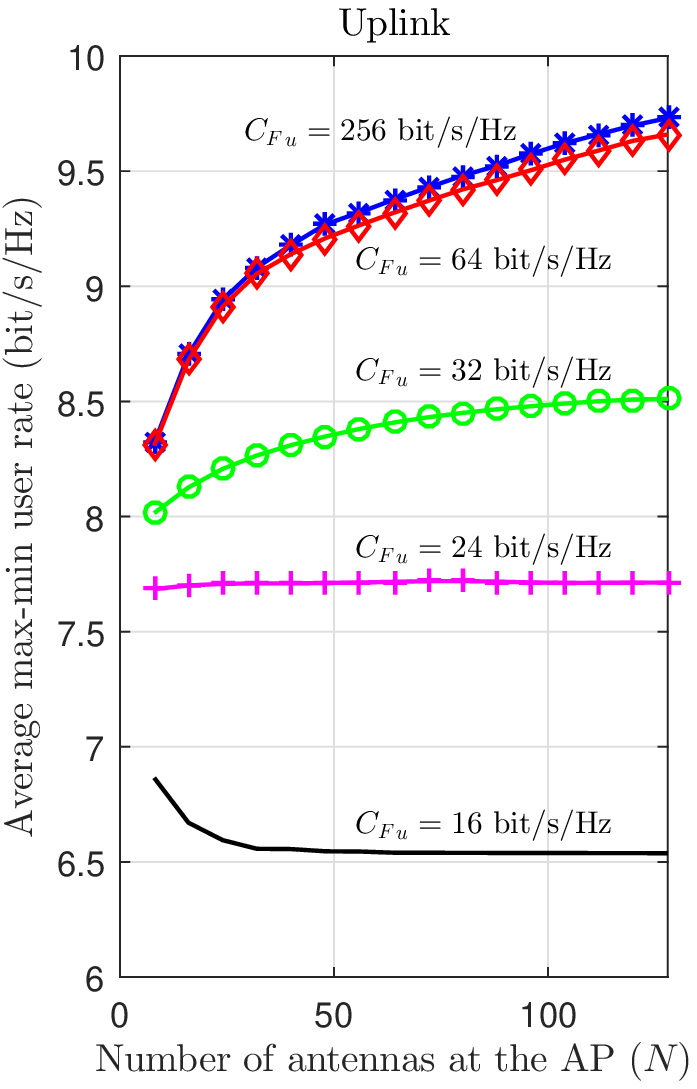}
%        \caption{\small Uplink}
        \label{fig:UL_FixK_varN_varCFu_min_rate}
    \end{subfigure}
%    \begin{subfigure}[t]{0.24\textwidth}
%        \centering
%        \includegraphics[height=7cm]{DL_FixK_varN_varCFd_Upsilon}
%%        \caption{\small Downlink}
%        \label{fig:DL_FixK_varN_varCFd_Upsilon}
%    \end{subfigure}
%    \begin{subfigure}[t]{0.24\textwidth}
%        \centering
%        \includegraphics[height=7cm]{UL_FixK_varN_varCFu_Omega}
%%        \caption{\small Uplink}
%        \label{fig:UL_FixK_varN_varCFu_Omega}
%    \end{subfigure}
%    \begin{subfigure}[t]{0.24\textwidth}
%        \centering
%        \includegraphics[height=7cm]{DL_FixK_varN_varCFd_sigmaqd2}
%%        \caption{\small Downlink}
%        \label{fig:DL_FixK_varN_varCFd_sigmaqd2}
%    \end{subfigure}
%    \begin{subfigure}[t]{0.24\textwidth}
%        \centering
%        \includegraphics[height=7cm]{UL_FixK_varN_varCFu_sigmaqu2}
%%        \caption{\small Uplink}
%        \label{fig:UL_FixK_varN_varCFu_sigmaqu2}
%    \end{subfigure}

    \caption{\small Average max-min rate per user \emph{versus} the number of antennas at the \glspl{AP} for different values of the fronthaul capacities ($K=20$ users, $L=8$ \gls{RF} chains, \gls{DCPA}).}
    \label{fig:FixK_varN_varCF_min_rate}
\end{figure}

\begin{figure}

    \centering
    \begin{subfigure}[t]{0.24\textwidth}
        \centering
        \includegraphics[height=7cm]{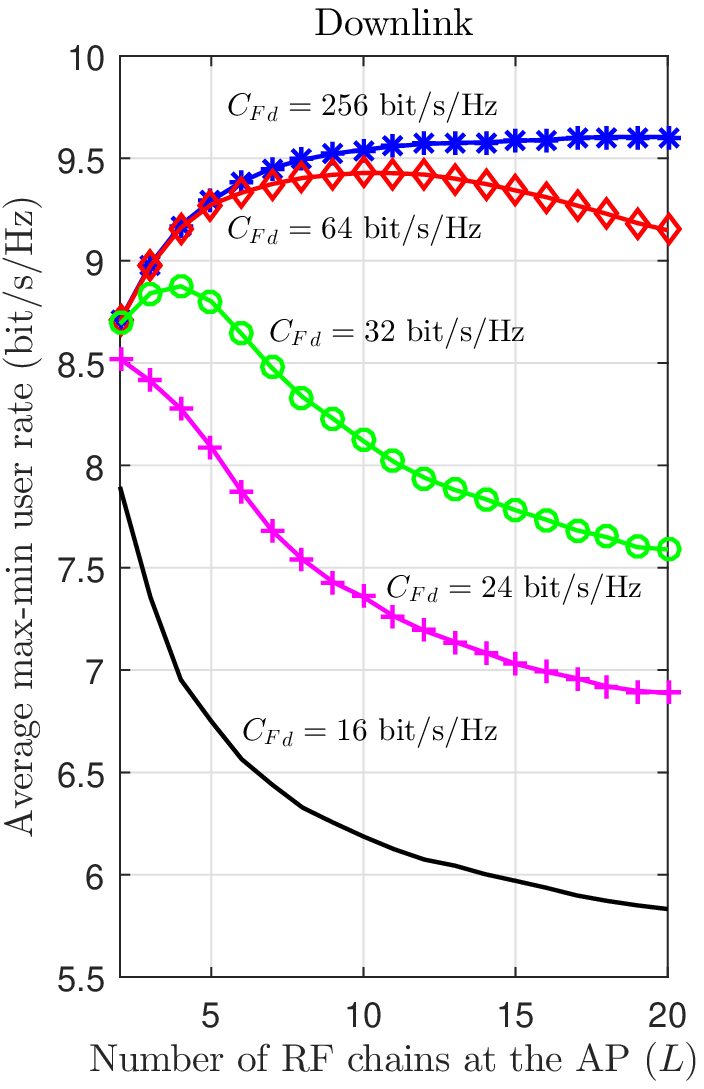}
%        \caption{\small Downlink}
        \label{fig:DL_FixK_varL_varCFd_min_rate}
    \end{subfigure}
    \begin{subfigure}[t]{0.24\textwidth}
        \centering
        \includegraphics[height=7cm]{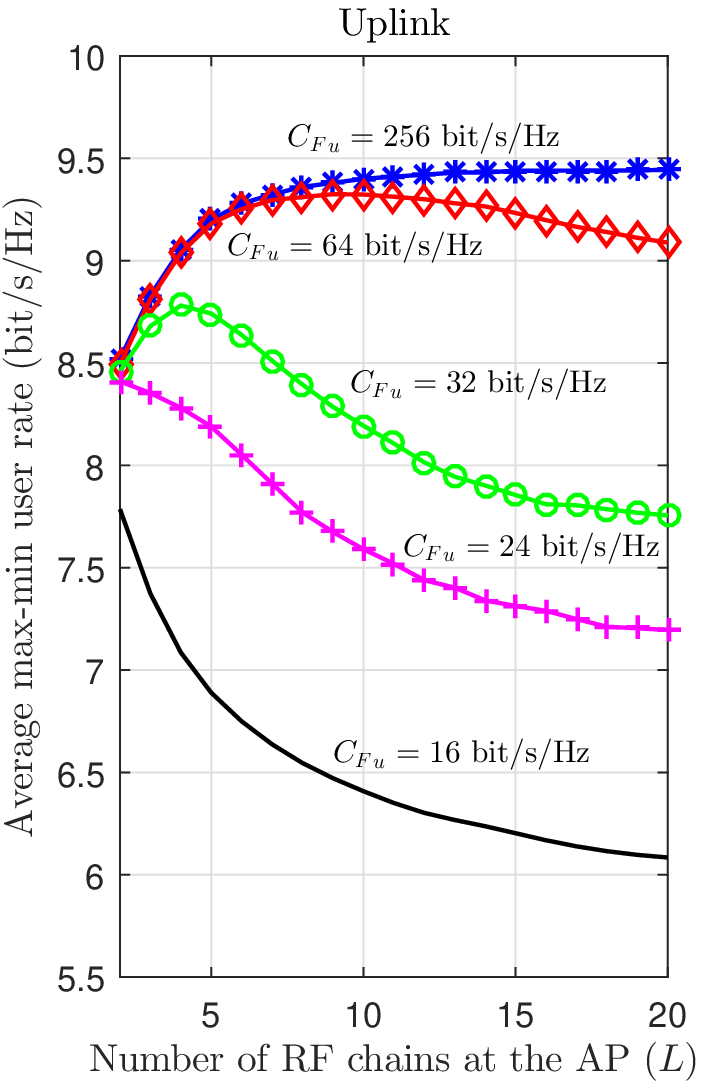}
%        \caption{\small Uplink}
        \label{fig:UL_FixK_varL_varCFu_min_rate}
    \end{subfigure}
%    \begin{subfigure}[t]{0.24\textwidth}
%        \centering
%        \includegraphics[height=7cm]{DL_FixK_varL_varCFd_Upsilon}
%%        \caption{\small Downlink}
%        \label{fig:DL_FixK_varL_varCFd_Upsilon}
%    \end{subfigure}
%    \begin{subfigure}[t]{0.24\textwidth}
%        \centering
%        \includegraphics[height=7cm]{UL_FixK_varL_varCFu_Omega}
%%        \caption{\small Uplink}
%        \label{fig:UL_FixK_varL_varCFu_Omega}
%    \end{subfigure}
%    \begin{subfigure}[t]{0.24\textwidth}
%        \centering
%        \includegraphics[height=7cm]{DL_FixK_varL_varCFd_sigmaqd2}
%%        \caption{\small Downlink}
%        \label{fig:DL_FixK_varL_varCFd_sigmaqd2}
%    \end{subfigure}
%    \begin{subfigure}[t]{0.24\textwidth}
%        \centering
%        \includegraphics[height=7cm]{UL_FixK_varL_varCFu_sigmaqu2}
%%        \caption{\small Uplink}
%        \label{fig:UL_FixK_varL_varCFu_sigmaqu2}
%    \end{subfigure}

    \caption{\small Average max-min rate per user \emph{versus} the number of \gls{RF} chains at the \glspl{AP} for different values of the fronthaul capacities ($K=20$ users, $N=64$ antennas, \gls{DCPA}).}
    \label{fig:FixK_varL_varCF_min_rate}
\end{figure}

\begin{figure*}

    \centering
    \begin{subfigure}[t]{0.48\textwidth}
        \centering
        \includegraphics[height=5.9cm]{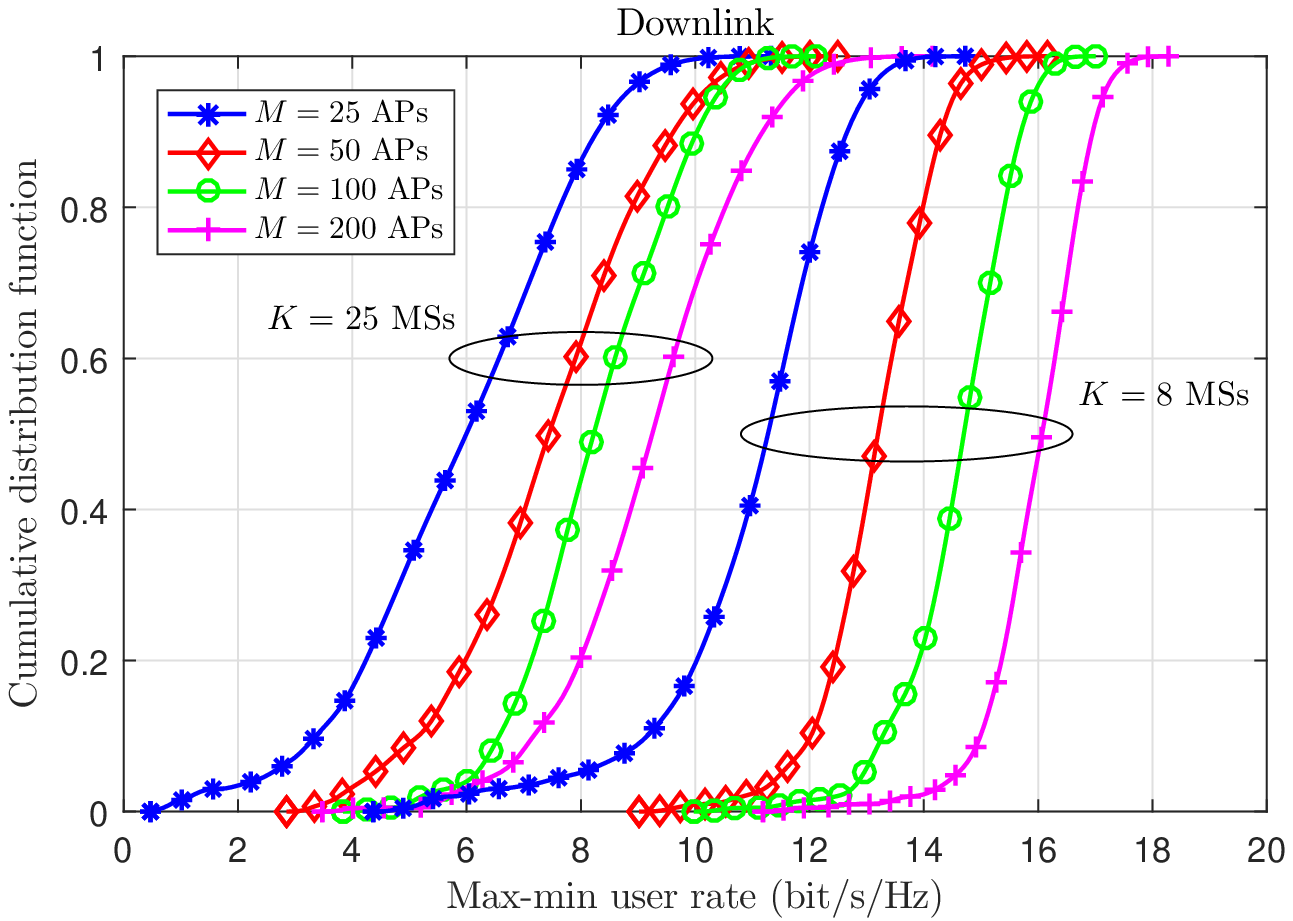}
%        \caption{\small Downlink}
        \label{fig:DL_CFR_VarM}
    \end{subfigure}
    \begin{subfigure}[t]{0.48\textwidth}
        \centering
        \includegraphics[height=5.9cm]{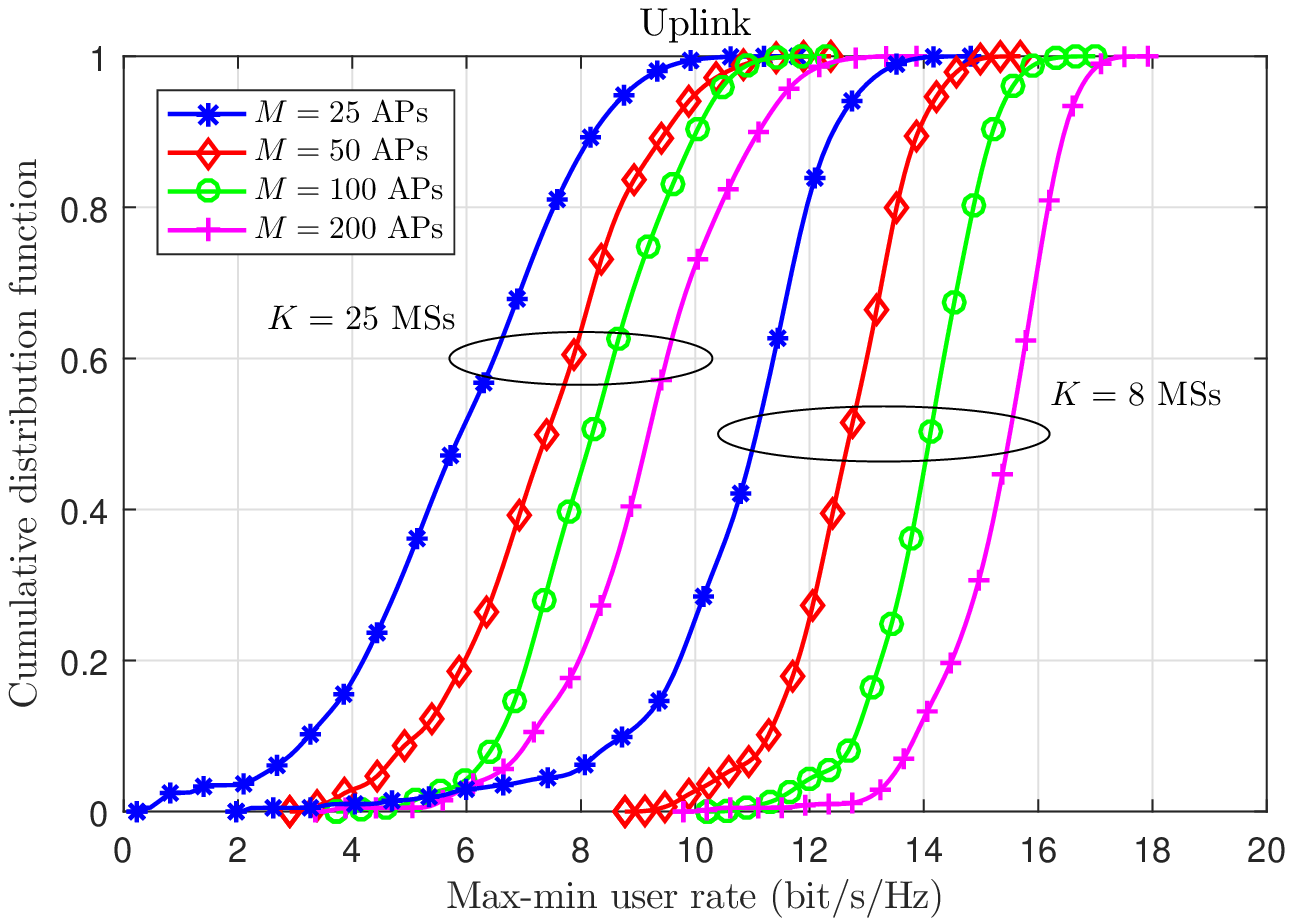}
%        \caption{\small Uplink}
        \label{fig:UL_CFR_VarM}
    \end{subfigure}
    \caption{\small CDF of the \gls{DL} and \gls{UL} achievable max-min rate per user for different values of the number of \glspl{AP} and active \glspl{MS} in the cell-free network ($N=64$ antennas, $L=8$ \gls{RF} chains, ${C_F}_d={C_F}_u=64$ bit/s/Hz, \gls{DCPA}).}
    \label{fig:CFR_VarM}
\end{figure*}

Results presented in Figs. \ref{fig:varCF_min_rate}, \ref{fig:varN_min_rate} and \ref{fig:varL_min_rate} were obtained assuming high-capacity fronthaul links with ${C_F}_d={C_F}_u=64$ bit/s/Hz. However, the amount of quantized data that has to be conveyed from (to) the \gls{CPU} to (from) the \glspl{AP} in the \gls{DL} (\gls{UL}) depends on the number of antennas and \gls{RF} chains at the \glspl{AP} (see Section \ref{sec:fronthaul_capacity}). Thus, in order to deepen in the study of the impact the \gls{RF} infrastructure may have on the achievable performance of the proposed cell-free \gls{mmWave} massive \gls{MIMO} system under constrained-capacity fronthaul links, the average max-min user rate is plotted in Figs. \ref{fig:FixK_varN_varCF_min_rate} and \ref{fig:FixK_varL_varCF_min_rate} against the number of antenna elements and \gls{RF} chains, respectively, for different values of the fronthaul capacities and assuming a fixed number of $K=20$ active \glspl{MS} in the network. In network setups using very high capacity fronthaul links (i.e., ${C_F}_d={C_F}_u=256$ bit/s/Hz), increasing the number of antenna elements $N$ and/or the number of \gls{RF} chains $L$ (up to $L=K$) is always beneficial as, in this case, the noise introduced by the quantization process is negligible and the system can take full advantage of the increased \gls{RF} resources. As the capacity of the fronthaul links decreases, however, the amount of noise introduced by the quantization process increases with both $N$ and $L$ and, therefore, a situation arises where the potential performance improvement provided by the increase of $N$ and/or $L$ is compromised by the performance reduction due to fronthaul capacity constraints. On the one hand, it can be observed in Fig. \ref{fig:FixK_varN_varCF_min_rate} that, for fixed numbers of users and \gls{RF} chains, there is a certain fronthaul capacity constraint value (near 24 bit/s/Hz in the setup used in this experiment) under which increasing the number of antenna elements at the array is counterproductive. On the other hand, results presented in Fig. \ref{fig:FixK_varL_varCF_min_rate} show that, for fixed numbers of users and antenna elements at the arrays, there is always an optimal number of \gls{RF} chains to be deployed (or activated) at the \glspl{AP} that is dependent on the capacity of the fronthaul links. In particular, for the network setups under consideration, the optimal number of \gls{RF} chains is equal to $L=10$, 4, and 1 when using fronthaul links with a capacity of 64 bit/s/Hz, 32 bit/s/Hz and less than 24 bit/s/Hz, respectively.

\subsection{Impact of the density of \glspl{AP}}

With the aim of evaluating the impact the density of \glspl{AP} per area unit may have on the performance of the proposed cell-free \gls{mmWave} massive \gls{MIMO} system, Fig. \ref{fig:CFR_VarM} represents the \gls{CDF} of the \gls{DL} and \gls{UL} achievable max-min user rate for different values of the number of \glspl{AP} in the network. It has been assumed in these experiments a fixed number of active \glspl{MS} equal to either $K=25$ or $K=8$ \glspl{MS}, the use of $L=8$ \gls{RF} chains fully-connected to a linear uniform antenna array with $N=64$ antenna elements, and the use of \gls{DL} and \gls{UL} fronthaul links with a capacity ${C_F}_d={C_F}_u=64$ bit/s/Hz. As expected, cell-free massive \gls{MIMO} scenarios with a high density of \glspl{AP} per area unit significantly outperform those with a low density of \glspl{AP} per area unit in both median and 95\%-likely achievable per-user rate performance. However, the achievable max-min user rate increase due to increasing the number of \glspl{AP} in the network is, again, subject to the law of diminishing returns. For instance, in scenarios with $K=25$ \glspl{MS}, the 95\%-likely achievable user rate is equal to 2.55, 4.33, 6.11 and 6.50 bit/s/Hz for cell-frre massive \gls{MIMO} networks with $M=25$, 50, 100 and 200 \glspl{AP}, respectively. That is, doubling the number of \glspl{AP} per area unit does not result in doubling the 95\%-likely achievable user rate. Similar conclusions can be drawn when looking at either the median or the average achievable user rates.

As was observed in results presented in previous subsections for high-capacity fronthaul setups, when the number of active users in the system is low, the achievable max-min rate values in the \gls{DL} are slightly higher than those achievable in the \gls{UL}. Instead, when the number of active users increases, the achievable max-min user rates are virtually identical in both the \gls{DL} and the \gls{UL}. Also, note that the dispersion of the achievable max-min user rates around the median tends to diminish as the density of \glspl{AP} increases. That is, cell-free massive \gls{MIMO} networks with a high density of \glspl{AP} per area unit tend to offer max-min achievable rates that suffer little variations irrespective of the location of the \glspl{AP} (i.e, irrespective of the scenario under evaluation).

\section{Conclusion}
\label{sec:Conclusion}

A novel analytical framework for the performance analysis of cell-free \gls{mmWave} massive \gls{MIMO} networks has been introduced in this paper. The proposed framework considers the use of low-complexity hybrid precoders/decoders where the \gls{RF} high-dimensionality phase shifter-based precoding/decoding stage is based on large-scale second-order channel statistics, while the low-dimensionality baseband multiuser \gls{MIMO} precoding/decoding stage can be easily implemented by standard ZF signal processing schemes using small-scale estimated \gls{CSI}. Furthermore, it also takes into account the impact of using capacity-constrained fronthaul links that assume the use of large-block lattice quantization codes able to approximate a Gaussian quantization noise distribution, which constitutes an upper bound to the performance attained under any practical quantization scheme. Max-min power allocation and fronthaul quantization optimization problems have been posed thanks to the development of mathematically tractable expressions for both the per-user achievable rates and the fronthaul capacity consumption. These optimization problems have been solved by combining the use of block coordinate descent methods with sequential linear optimization programs. Results have shown that the proposed \gls{DCPA} suboptimal pilot allocation strategy, which is based on the idea of clustering by dissimilarity, overcomes the computational burden of the optimal small-scale CSI-based pilot allocation scheme while clearly outperforming the pure random and balanced random schemes. It has also been shown that, although increasing the fronthaul capacity and/or the density of \glspl{AP} per area unit is always beneficial from the point of view of the achievable max-min user rate, the marginal increment of performance produced by each new increment of these parameters suffers from the law of diminishing returns, especially for network setups with a high number of active \glspl{MS}. Moreover, simulation results indicate that, as the capacity of the fronthaul links decreases, the potential performance improvement provided by the increase of the number of antenna elements $N$ and/or the number of \gls{RF} chains $L$ is compromised by the performance reduction due to the corresponding increase of the fronthaul quantization noise. In particular, for fixed numbers of users and \gls{RF} chains, there is a certain fronthaul capacity constraint value (near 24 bit/s/Hz in the setups under consideration) under which increasing the number of antenna elements at the array is counterproductive. Similarly, for fixed numbers of users and antenna elements at the arrays, there is always an optimal number of \gls{RF} chains to be deployed (or activated) at the \glspl{AP} that is dependent on the capacity of the fronthaul links. For future work, it would be interesting to develop low-complexity pilot- and power-allocation techniques specifically designed to maximize the energy efficiency of cell-free \gls{mmWave} massive \gls{MIMO} networks considering both the fronthaul capacity constraints and the fronthaul power consumption. It would also be interesting to explore the use of partially-connected \gls{RF} precoding/decoding architectures and the implementation of baseband \gls{MU-MIMO} precoding/decoding other than the \gls{ZF} scheme.

\appendices

\section{Proof of Theorem \ref{theo:theorem_DL}}
\label{app:Appendix_1}

Following an approach similar to that proposed by Nayebi \emph{et al.} in \cite{Nayebi17}, the signal received by the $k$th \gls{MS} in \eqref{eq:ydk} can be rewritten as ${y_d}_k={y_d}_{k\,0} + {y_d}_{k\,1} + {y_d}_{k\,2} + {n_d}_k$, where the useful, interuser interference, and quantization noise terms can be expressed as ${y_d}_{k\,0}=\sqrt{\upsilon_k} {s_d}_k$, ${y_d}_{k\,1}=\tilde{\bs{g}}_k^T \bs{W}_d^{BB} \bs{\Upsilon}^{1/2}\bs{s}_d$, and ${y_d}_{k\,2} = \bs{g}_k^T \bs{q}_d=\sum_{m=1}^M \bs{g}_{km}^T {\bs{q}_d}_m$, respectively. Now, considering that data symbols, quantization noise, thermal noise, and channel-related coefficients are mutually independent, the terms ${y_d}_{k\,0}$, ${y_d}_{k\,1}$, ${y_d}_{k\,2}$ and ${n_d}_k$ are mutually uncorrelated and thus, based on the worst-case uncorrelated additive noise \cite{Hassibi03}, the achievable \gls{DL} rate for user $k$ is lower bounded by ${R_d}_k=\log_2\left(1+{\SINR_d}_k\right)$, with
\begin{equation*}
   {\SINR_d}_k=\frac{\mathbb{E}\left\{\left|{y_d}_{k\,0}\right|^2\right\}}{\mathbb{E}\left\{\left|{y_d}_{k\,1}\right|^2\right\}+\mathbb{E}\left\{\left|{y_d}_{k\,2}\right|^2\right\}+\sigma_d^2},
\end{equation*}
where $\mathbb{E}\left\{\left|{y_d}_{k\,0}\right|^2\right\}=\upsilon_k$,
\begin{equation*}
\begin{split}
   &\mathbb{E}\left\{\left|{y_d}_{k\,1}\right|^2\right\}=\mathbb{E}\left\{\bs{s}_d^H\bs{\Upsilon}^{1/2} {\bs{W}_d^{BB}}^H \tilde{\bs{g}}_k^* \tilde{\bs{g}}_k^T \bs{W}_d^{BB} \bs{\Upsilon}^{1/2}\bs{s}_d\right\} \\
   &\qquad=\tr\left(\bs{\Upsilon} \mathbb{E}\left\{{\bs{W}_d^{BB}}^H \tilde{\bs{g}}_k^* \tilde{\bs{g}}_k^T \bs{W}_d^{BB}\right\}\right) \\
   &\qquad=\sum_{k'=1}^K \upsilon_{k'} \left[\diag\left(\mathbb{E}\left\{{\bs{W}_d^{BB}}^H \tilde{\bs{g}}_k^* \tilde{\bs{g}}_k^T \bs{W}_d^{BB}\right\}\right)\right]_{k'},
\end{split}
\end{equation*}
and
\begin{equation*}
\begin{split}
   \mathbb{E}\left\{\left|{y_d}_{k\,2}\right|^2\right\}&=\sum_{m=1}^M \mathbb{E}\left\{{\bs{q}_d}_m^H \bs{g}_{km}^* \bs{g}_{km}^T {\bs{q}_d}_m\right\} \\
                                                       &=\sum_{m=1}^M {\sigma_q^2}_{dm} \tr\left(\bs{R}_{mk}^{RF}\right).
\end{split}
\end{equation*}

\section{Proof of Theorem \ref{theo:theorem_UL}}
\label{app:Appendix_2}

The detected signal at the \gls{CPU} corresponding to the symbol transmitted by the $k$th \gls{MS} in \eqref{eq:yuk} can be rewritten as ${y_u}_k={y_u}_{k\,0} + {y_u}_{k\,1} + {y_u}_{k\,2} + {y_u}_{k\,3}$, where the useful, interuser interference, quantization noise and thermal noise terms can be expressed as ${y_u}_{k\,0}=\sqrt{P_u} \sqrt{\omega_k} {s_u}_k$, ${y_u}_{k\,1}=\sqrt{P_u}\left[\bs{W}_u^{BB} \tilde{\bs{G}} \bs{\Omega}^{1/2} \bs{s}_u\right]_k$, ${y_u}_{k\,2} = \left[\bs{W}_u^{BB} \bs{q}_u \right]_k$, and ${y_u}_{k\,3} = \left[\bs{W}_u^{BB} \bs{n}_u\right]_k$, respectively. As in the \gls{DL}, since data symbols, quantization noise, thermal noise, and channel-related coefficients are mutually independent, the terms ${y_u}_{k\,0}$, ${y_u}_{k\,1}$, ${y_d}_{k\,2}$ and ${y_d}_{k\,3}$ are mutually uncorrelated and thus, based on the worst-case uncorrelated additive noise \cite{Hassibi03}, the achievable \gls{UL} rate for user $k$ is lower bounded by ${R_u}_k=\log_2\left(1+{\SINR_u}_k\right)$, with
\begin{equation*}
   {\SINR_u}_k=\frac{\mathbb{E}\left\{\left|{y_d}_{k\,0}\right|^2\right\}}{\mathbb{E}\left\{\left|{y_d}_{k\,1}\right|^2\right\}+\mathbb{E}\left\{\left|{y_d}_{k\,2}\right|^2\right\}+\mathbb{E}\left\{\left|{y_d}_{k\,3}\right|^2\right\}},
\end{equation*}
where $\mathbb{E}\left\{\left|{y_u}_{k\,0}\right|^2\right\}=P_u \omega_k$,
\begin{equation*}
\begin{split}
   &\mathbb{E}\left\{\left|{y_u}_{k\,1}\right|^2\right\}=P_u\mathbb{E}\left\{\bs{s}_u^H\bs{\Omega}^{1/2} \tilde{\bs{G}}^H {\bs{w}_{uk}^{BB}}^H \bs{w}_{uk}^{BB} \tilde{\bs{G}}\bs{\Omega}^{1/2} \bs{s}_u\right\} \\
   &\qquad=P_u \tr\left(\bs{\Omega} \mathbb{E}\left\{\tilde{\bs{G}}^H {\bs{w}_{uk}^{BB}}^H \bs{w}_{uk}^{BB} \tilde{\bs{G}}\right\}\right) \\
   &\qquad=P_u \sum_{k'=1}^K \omega_{k'} \left[\diag\left(\mathbb{E}\left\{\tilde{\bs{G}}^H {\bs{w}_{uk}^{BB}}^H \bs{w}_{uk}^{BB} \tilde{\bs{G}}\right\}\right)\right]_{k'},
\end{split}
\end{equation*}
with $\bs{w}_{uk}^{BB}$ denoting the $k$th row of $\bs{W}_u^{BB}$, or, equivalently,
\begin{equation*}
\begin{split}
   &\mathbb{E}\left\{\left|{y_u}_{k\,1}\right|^2\right\} = P_u\left[\diag\left(\mathbb{E}\left\{\bs{W}_u^{BB} \tilde{\bs{G}} \bs{\Omega} \tilde{\bs{G}}^H {\bs{W}_u^{BB}}^H \right\}\right)\right]_k \\
   &\qquad=P_u \sum_{k'=1}^K \omega_{k'} \left[\diag\left(\mathbb{E}\left\{\bs{W}_u^{BB} \tilde{\bs{g}}_{k'} \tilde{\bs{g}}_{k'}^H {\bs{W}_u^{BB}}^H \right\}\right)\right]_k,
\end{split}
\end{equation*}
and, finally,
\begin{equation*}
\begin{split}
   &\mathbb{E}\left\{\left|{y_u}_{k\,2}\right|^2\right\}=\left[\diag\left(\mathbb{E}\left\{\bs{W}_u^{BB} \bs{q}_u \bs{q}_u^H {\bs{W}_u^{BB}}^H\right\}\right)\right]_k \\
   &\qquad=\sum_{m=1}^M \left[\diag\left(\mathbb{E}\left\{\bs{W}_u^{BB} \bs{q}_{um} \bs{q}_{um}^H {\bs{W}_u^{BB}}^H\right\}\right)\right]_k \\
   &\qquad=\sum_{m=1}^M {\sigma_q^2}_{um}\left[\diag\left(\mathbb{E}\left\{\bs{W}_{u\,m}^{BB} {\bs{W}_{u\,m}^{BB}}^H\right\}\right)\right]_k,
\end{split}
\end{equation*}
and, analogously,
\begin{equation*}
\begin{split}
   \mathbb{E}\left\{\left|{y_d}_{k\,3}\right|^2\right\}=\sigma_u^2(N)\sum_{m=1}^M \left[\diag\left(\mathbb{E}\left\{\bs{W}_{u\,m}^{BB} {\bs{W}_{u\,m}^{BB}}^H\right\}\right)\right]_k.
\end{split}
\end{equation*}

\bibliographystyle{IEEEtran}
\bibliography{Cell_free_Hybrid_precoding}

\end{document}